%% file: Epidemics_2.tex
\documentclass[envcountchap,sectrefs]{svmono}

\usepackage{mathptmx}
\usepackage{helvet}
\usepackage{courier}
\usepackage{type1cm}

\usepackage{makeidx}         
\usepackage{graphicx}        
\usepackage{multicol}        
\usepackage[bottom]{footmisc}

\usepackage[bb=boondox]{mathalfa}

\usepackage{graphicx}
\usepackage{epsfig}

\usepackage{hyperref} 

\usepackage{bm}

\usepackage{bbm}
\usepackage{amssymb}
\usepackage{bm}

\newcommand{\One}{\mathbbm{1}}
\newcommand{\PP}{\mathbb{P}}
\newcommand{\brw}{\boldsymbol{\mathrm{w}}}

\usepackage{xcolor}
\usepackage{fancybox}

\usepackage{tcolorbox} 

\usepackage{amssymb}
\usepackage{amsfonts}
\usepackage{amsmath}

\newtheorem{Teorema}{Theorem}
\newtheorem{Lema}{Lemma}

\newtheorem{Definicao}{Definition}
\newtheorem{Obs}{Remark}

\newtheorem{ass}{Assumption}

\newtheorem{thm}{\protect\theoremname}
  \newtheorem{prop}[thm]{\protect\propositionname}

\def\bm#1{\textbf{\em #1}}

\makeindex             

\begin{document}

\tcbset{
  colframe=gray!110, 
  colback=yellow!10,     
  coltext=black,         
  boxrule=1mm,           
  arc=3mm,               
  width=\linewidth,      
  left=2mm,              
  right=2mm              
}


\frontmatter

\include{first}

\include{preface}

\tableofcontents

\mainmatter

\newpage


\chapter{Classical compartimental models}

\noindent\hfill
\begin{minipage}[t]{9cm}
{\small
We discuss the basic results of the SIS and SIR models. For the SIS model, we show that when the reproduction number is larger than 1, the endemic equilibrium is stable. For the SIR model, we derive the peak of the number of infections and the herd immunity threshold.}
\end{minipage}

\section{Infectious Diseases and Isolation}

Around 400 infectious diseases have been identified since 1940. New pathogens are emerging at higher rates despite increased awareness and surveillance \cite{dassouki2021}. A major public health concern is when and how the next outbreak will occur \cite{Woolhouse2016}. The Spanish flu of 1917, which killed 50 million people, was the worst pandemic ever recorded, and it occurred at a time when ship travel was the fastest mode of transportation worldwide. In today's highly connected world, an epidemic can travel at high speed. Indeed, swine flu was first detected in April 2009 in Mexico, and within a week, it appeared in London \cite{dawood2009}.

Mathematical models help to understand the impact of traveling in a tightly connected environment \cite{colizza2006}, where the spreading also affects different social groups, such as income-related status and age.  An example is the early spread of COVID-19. Italy and Germany provide examples of how societal, demographic, and policy differences can impact the spread of an epidemic. 

Italy has high levels of intergenerational contact and experienced a severe outbreak early in the pandemic. Germany, on the other hand, implemented effective early interventions and benefited from fewer intergenerational interactions, resulting in comparatively lower case fatality rates. In fact, while Germany had a fatality rate of about $6\%$ in a population aged 60-79, Italy had a fatality rate of about 12\% in the same age group \cite{dowd2020,mckinsey2020}. Moreover, while South Korea had about 21\% of cases in the population aged 60+, Italy had around 60\% of cases in this age group \cite{backhaus2020}. Over the course of the pandemic, however, Germany also suffered from widespread infections, while South Korea was fairly able to keep track of the infection numbers via several control policies. 

The contrasting experiences of South Korea, Germany, and Italy during the early stages of the COVID-19 pandemic offer valuable insights into the importance of public health interventions and demographic differences. 

Italy's extensive intergenerational interactions in multigenerational households contributed to higher transmission rates and a larger burden of severe cases. The delay in implementing widespread testing and containment measures further exacerbated the crisis. The differences between these two countries highlight the critical role of early interventions, healthcare system capacity, and demographic factors in shaping the outcomes of a pandemic.

In both cases, understanding the impact of travelling and interactions between different age groups, compartmental models can provide insights into the spread of epidemics. Compartimental models in epidemiology divide a population into groups based on epidemiological status, such as susceptible, infectious, recovered, or those with specific characteristics, including age and geographical location. The goal is then to track how individuals transition between these groups over time, helping to understand disease spread and the impact of interventions. This allows us to understand how age groups are affected and how epidemics spread across cities, as well as the role of urban mobility.

\section{The SIS Model}

During the spread of an epidemic, the population is divided into groups. For example, the infected, non-infected, immune, and so on. Some infections, such as the flu, do not confer long-lasting immunity. After recovering from the infection, individuals become susceptible again. This applies equally well to computer viruses.
In this scenario, the population is divided into two classes:

\begin{description}
\item[{\it S}] -- Healthy individuals who can contract the disease.
\item[{\it I}] -- Individuals who have contracted the disease and are now infected.
\end{description}

\noindent
Throughout the text, $S$ and $I$ will denote the fraction of individuals in the population. Let us make the following assumptions:

\begin{itemize}
\item[] {
H1} -- When an infected individual meets a susceptible individual, the susceptible individual becomes infected with a certain probability. \\
\item[] {
H2} -- After some time, an infected individual becomes healthy and returns to the group of susceptibles, that is, it can get reinfected. \\
\item[] {
H3} -- There are no deaths. The number of individuals remains constant over time.
\end{itemize}

\noindent
Under the above assumptions, we derive the following system:
\begin{eqnarray}
\text{Change in } S &\propto& - \# \text{infected individuals} + \# \text{recovered individuals} \nonumber \\\text{Change in } I &\propto& + \# \text{infected individuals} - \# \text{recovered individuals} \nonumber
\end{eqnarray}

\noindent
From hypothesis {\rm H1}, we have
\begin{eqnarray}
\# \text{infected individuals} &=& \beta S I \nonumber
\end{eqnarray}

Here, $\beta$ represents the transmission rate. To create a new infected individual, one needs both a susceptible and an infected
\noindent
From hypothesis {\rm H2}, the fraction of individuals who become healthy is proportional to $I$:
\begin{eqnarray}
\# \text{recovered individuals} &=& \gamma I \nonumber
\end{eqnarray}

\noindent
and $\gamma$ represents the recovery rate of the disease. Thus, we obtain the SIS model:
\begin{eqnarray}
\dot{S} &=& - \beta S I + \gamma I \\\dot{I} &=& + \beta S I - \gamma I
\end{eqnarray}

\noindent
We will study some properties of the SIS model. Note $ \dot{S} + \dot{I} = 0$ implying $S + I = \text{constant}$, 
which satisfies {\rm H3}. We normalize the total number of individuals:
$$
S + I = 1 \,\,\,\, \Rightarrow \,\,\, S = 1 - I
$$
and substituting into the equation for $I$, we obtain a differential equation for $I$:
$$
\dot{I} = (\beta - \gamma) I - \beta I^2
$$
This equation governs the proportion of infected individuals and appears in multiple areas, such as ecology and population dynamics. It can be fully solved, but first, let's look at its equilibrium states, that is, 
values of $I$ that don't change over time.

{\it Equilibrium States:} Let us analyze the equilibrium solutions of the differential equation, i.e., solutions such that
$
I = \text{constant}.
$
These correspond to
$$
I [(\beta - \gamma) - \beta I] = 0
$$
Thus, we have two solutions:
$$
I_{\text{free}} = 0
$$
which corresponds to the case when the disease is eradicated, and
$$
I_{\text{end}} = 1 - \frac{1}{r_0}, \mbox{ ~ where ~ }
r_0 = \frac{\beta}{\gamma} ~\text{is the basic reproduction number}
$$
Since $I > 0$, the endemic equilibrium solution $I_{\text{end}}> 0$ exists only when
$$
r_0 > 1
$$

\section{Biological vs. effective reproduction numbers}

We define
\[
r_0 := \frac{\beta}{\gamma},
\]
where $\beta$ is the transmission rate and $\gamma$ is the recovery rate.
The quantity $r_0$ is the \emph{biological reproduction number}: it represents the
expected number of secondary infections produced by a single infected individual
in a fully susceptible, well--mixed population. As such, $r_0$ is mainly determined
by biological and individual-level factors, such as infectiousness, duration of
infection, and basic behavioral patterns.

In more realistic settings, as we will explore in these lectures, the disease transmission is not governed by biology alone. Infrastructure, mobility patterns, social interactions, and heterogeneous contact structures play a crucial role in shaping how infections spread through a population. These effects cannot be captured by $r_0$ alone.

To account for these additional mechanisms, we introduce an \emph{effective reproduction number}, denoted by $R_0$. This quantity incorporates both the biological transmissibility of the disease and the structure of interactions in the population. In general, it can be written schematically as
\[
R_0 = r_0 \times (\text{structural amplification factor}),
\]
where the amplification factor depends on features such as network connectivity, mobility between subpopulations, or age-dependent contact patterns.

In this sense, while $r_0$ reflects intrinsic biological properties of the disease, the effective reproduction number $R_0$ quantifies how social organization and infrastructure amplify or suppress transmission. Mathematically, $R_0$ typically appears as the \emph{dominant eigenvalue} of a linear operator describing how infections generate new infections during the early stage of an epidemic. 

At the beginning of an outbreak, the number of infected individuals is small, and the dynamics can be well approximated by a linear system. In this regime, each infected individual produces new infections according to the linear operator. The dominant eigenvalue of this operator measures the maximal average amplification of infections across the population. In this lecture, since we don't have further structure in the population $R_0 = r_0$.

\section{Main Lessons}

Lets discuss four main interesting feature of the spreading

\begin{tcolorbox}
\begin{center}
{\it Lesson 1: In the SIS model, disease is endemic when $r_0>1$}
\end{center}
\end{tcolorbox}

When $r_0 < 1$, all solutions of the SIS model converge to $I_{\text{free}}$, meaning that the disease is always eradicated. On the other hand, when $r_0 > 1$, the equilibrium $I_{\text{free}}$ is unstable; even starting with a very small number of infected individuals, the solutions converge to the endemic equilibrium $I_{\text{end}}$.

{\it Local Stability of the Equilibrium States:} Recall that given a scalar ordinary differential equation
$$
x^{\prime}= f(x) 
$$
where $f(x_0)=0$ and $f$ is smooth then the equilibrium $x=x_0$ is exponentially stable 
if the coefficient $\mu = f^{\prime}(x_0)<0$. A more general result is proven in the appendix.

To analyze the local stability, we linearize the system around the equilibrium points. At \(I^* = 0\):
coefficient is

\begin{equation}
\mu = (\beta - \gamma)
\end{equation}

Hence, if \(\beta < \gamma\), then the disease-free equilibrium is locally stable. However,  when \(\beta > \gamma\), then the disease-free equilibrium is unstable.

At \(I^* = 1 - \frac{1}{r_0}\) the coefficient is
\begin{equation}
\mu = \frac{d}{dI} \left( (\beta - \gamma) I - \beta I^2 \right)\Big|_{I=I_{\rm end}} = - (\beta - \gamma)
\end{equation}
Thus, the endemic equilibrium is locally stable if it exists (\(\beta > \gamma\)).

In other words, the disease-free equilibrium is stable if the basic reproduction number \(r_0 = \frac{\beta}{\gamma} < 1\). The endemic equilibrium is stable if \(r_0 > 1\). \\

~\\

{\it Effects of Seasonality:}
Seasonality can influence the dynamics of the SIS model when the transmission rate \( \beta \) is modeled as a periodic function of time, reflecting seasonal variations in contact rates or environmental conditions. The  SIS model is:
\begin{equation}
\frac{dI}{dt} = \beta(t) I (1 - I) - \gamma I,
\end{equation}
with \( \beta(t+T) = \beta(t) \) for a $T>0$ called period. A common choice is \( \beta (t) = \beta_0 (1 + \epsilon \cos(\omega t)) \), where  \( \beta_0 \) is the average transmission rate,
     \( \epsilon \) is the amplitude of seasonal variation, with \( 0 \leq \epsilon \leq 1 \)
%
%
This time-dependent SIS model has an explicit solution; see the appendix. In fact, when
$$
\int_{0}^T \beta(s) ds < \gamma
$$
then the equilibrium is stable.

\subsection{The SIR Model}

The SIS model assumes that individuals who recover can be reinfected. While this assumption is valid for variations of the flu or computer viruses, in many cases, once an individual recovers, they acquire immunity. In such cases, we need to modify the assumptions. Now, in addition to $S$ and $I$, we introduce a new class:

\noindent
{\it R} -- Recovered: Removed from the population. They had the disease and recovered, are now immune, or are isolated until recovery or death.

This model is reasonably predictive for infectious diseases transmitted from human to human, where recovery confers lasting immunity, such as measles, mumps, and rubella. As before, we have:
\begin{eqnarray}
\text{Changes in } S &\propto& - \# \text{infected individuals} \nonumber \\
\text{Changes in } \,I\, &\propto& + \# \text{infected individuals} - \# \text{recovered individuals} \nonumber \\
\text{Changes in } R &\propto& + \# \text{infected individuals} \nonumber
\end{eqnarray}

This leads to the following equations
\begin{eqnarray}
\dot{S} &=& - \beta S I \nonumber \\
\dot{I} &=&   + \beta S I - \gamma I \\
\dot{R} &=& \gamma I \nonumber
\end{eqnarray}

Notice that, similar to the SIS model, here we also have $S+I+R=1$.

In our small digression at the end of the last section, we noted that the SIS model can be explicitly solved. Here, an explicit solution is not possible. However, the asymptotic dynamics of the SIR model are simpler.


\begin{tcolorbox}
\begin{center}
{\it Lesson 2: In the SIR model, the epidemic always ends}
\end{center}
\end{tcolorbox}
Since $\dot S \le 0$ and  
as $S$ is monotonic and positive,  
$$
\lim_{t \rightarrow \infty} S(t) = S_{\infty}
$$

On the other hand, the asymptotic value of the number of infected individuals is  
$$
\lim_{t \rightarrow \infty} I(t) = 0
$$
To see this, we note that  
\begin{eqnarray}
\int_0^{\infty} \dot S(t) dt &=&  - \beta \int_0^{\infty} S(t) I(t) dt \\
S_0 - S_{\infty} &=& \beta \int_0^{\infty} S(t) I(t) dt
\end{eqnarray}
and since $S(t) \ge S_{\infty}$,  
$$
S_0 - S_{\infty} \ge \beta S_{\infty} \int_0^{\infty} I(t) dt
$$
This shows that $\int_0^{\infty} I(t) dt$ is integrable, and hence $I(t) \rightarrow 0$ as $t \rightarrow \infty$.  

The point of interest here is not the long-term behavior but rather understanding the maximum value of $I$ to help authorities prepare for the disease. \\

\begin{tcolorbox}
\begin{center}
{\it Lesson 3: Proportion of infected does not grow when $r_0<1$}
\end{center}
\end{tcolorbox}

To determine when the number of infected individuals increases or decreases, we analyze the equation 
\[
I^{\prime} = I \left( \beta S - \gamma \right).
\]

Recall that \(I(t) > 0\) cannot be negative, thus, the sign of \(\frac{dI}{dt}\) depends on the term \(\beta S - \gamma\). For $I$ to grow we need \(I^{\prime} > 0\) thus 
\[
\beta S - \gamma > 0 \quad \implies  \quad S >  \frac{1}{r_0} > 1
\]
which is not possible.

{\it Small infections die out exponentially fast}
Consider
\[
   (S^*, I^*, R^*) = (1, 0, 0),
   \]

We linearize the system around this equilibrium.
\[
S = 1 - x, \quad I = y, \quad R = z,
\]
The linearized system is:
\[
\frac{d y }{dt} = (\beta - \gamma) y 
\]
and since 
\[
\beta < \gamma \quad \implies \quad r_0 = \frac{\beta}{\gamma} < 1.
\]
we obtain $I(t) \approx I_0 e^{(\beta - \gamma) t}$.

\newpage

\begin{tcolorbox}
\begin{center}
{\it Lesson 4: The peak of the disease depends on $r_0$.}
\end{center}
\end{tcolorbox}

To determine the peak number of infected individuals, we use the following trick. Considering only the equations for $S$ and $I$, and dividing one by the other, we obtain  
$$
\frac{dI}{dS} = -1 + \frac{1}{r_0} \frac{1}{S}
\,\,\, \Rightarrow  \,\,\,\, I + S - \frac{1}{r_0} \ln S = C
$$
where $C$ is a constant. Thus, for all time, the evolution of $S$ and $I$ respects the above condition. To determine the value of $I_{\text{max}}$, we set  
$$
\frac{dI(S^*)}{dS} = 0 \Rightarrow S^* = \frac{1}{r_0}
$$
and hence  
$$
I_{\text{max}} = - \frac{1}{r_0} + \frac{1}{r_0} \ln \frac{1}{r_0} + I_0 + S_0 + \frac{1}{r_0} \ln S_0
$$

Let us focus on the case where $S(0) \approx 1$ and $I(0) \approx 0$, and we obtain  
$$
I_{\text{max}} \approx 1 - \frac{1}{r_0} + \frac{1}{r_0} \ln \frac{1}{r_0} = I_{\rm end} +\frac{1}{r_0} \ln \frac{1}{r_0} 
$$

~\\
\begin{tcolorbox}
\begin{center}
{\it Lesson 5: When enough people are infected, the spread curbs by itself}
\end{center}
\end{tcolorbox}

Notice that $S$ is non-increasing 
$$S(t) \le S(0) \le p, \text{for all } t\ge 0.$$
Next, from the equation for $I$, we have
\[
I'(t) = \gamma (r_0 S(t) - 1)\, I(t)  \le 0 \,\,\,\mbox{~ when ~ } r_0 p -1\le 0
\]
that leads to \[
p < \frac{1}{r}, 
\]
then the infective class $I(t)$ starts decreasing immediately and no epidemic outbreak occurs. This is the herd immunity condition in the SIR model. Thus, once a fraction $1/r$ of the population is infected, the fraction of infected will decay exponentially fast. Notice that this value is usually very large. For instance, during COVID-19, we had $r\approx 2.5$, which would lead to $p\approx 40\%$ of susceptibles remaining, that is, $60\%$ of the population would need to be infected.

\chapter{Lecture 3: Fitting SIR models to Data}

\noindent\hfill
\begin{minipage}[t]{9cm}
{\small
We present three illustrations of how we can obtain the SIR parameters from data. First, for the Black Death, we obtain the parameters from the invariant relating $S$ and $I$ that we found previously. Second, we fit the SIR model to the spread of H1N1 in a boarding school in the UK in 1978.}
\end{minipage}


\section{The Black Death}

The village of Derbyshire, in Eyam, England, experienced an outbreak of bubonic plague in 1665-1666. The village is known as the "plague village" for choosing to isolate itself when the plague was discovered there in August 1665. Records of the disease's spread have been preserved. The initial population of Eyam was 350. By mid-May 1666, nine months after the outbreak began, there were 254 susceptible individuals and seven infectious individuals. By October 20th, there were 83 susceptible individuals and zero infectious individuals. The infectious period of bubonic plague is 11 days. To estimate the reproduction number $r_0$, we use the equation we deduced above
$$
I(t) + S(t) +\frac{1}{r_0} \ln S(t) = I_0 + S_0 +\frac{1}{r_0} \ln S_0 
$$  
Taking the limit $t \rightarrow \infty$ and that $I_{\infty} = 0$ along with $I_0 \approx 0$
we can estimate the reproduction number as   
$$
r_0 = \frac{\ln\left( \frac{S_0}{S_{\infty}}\right)}{S_0 - S_{\infty}} \approx 1.89
$$

That is, when the epidemic is over, it is easy to estimate the reproduction number.

\section{Spread of Influenza A}

In January and February of 1978, an influenza epidemic occurred in a boarding school in northern England that housed 763 boys. The boys had returned from their Christmas vacations from various locations around the world. One boy who had returned from Hong Kong exhibited a high fever between January 15th and 18th. By January 22nd, three boys were sick. The table below shows the number of sick boys from the ninth day onwards, starting from January 22nd ($n = 1$).

\begin{center}
\begin{tabular}{cc}
\hline
Day & \# Infected\\
\hline
\hline
3 & 25 \\
4 & 75 \\
5 & 227 \\
6 & 296 \\
7 & 258 \\
8 & 236 \\
9 & 192\\
10 & 126\\
11 & 71\\
12 & 28\\
13 & 11\\
14 & 7\\
\hline
\end{tabular}
\end{center}

The number of boys who escaped influenza was 19. The average duration of illness was 5 days. However, they remained infectious, i.e., capable of transmitting the disease, for about 2 days. An examination revealed that they were infected with the Influenza A H1N1 virus.

A relevant problem is predicting the maximum number of individuals who will be infected. We aim to make this prediction using only the data available at the beginning of the epidemic. For this, we need to estimate $\beta$ and $\gamma$ from the data.

~

\noindent
{\it Method 1: Estimating the peak using early recovered individuals.}  
We use the same argument as before by assuming $S \approx 1$ and $I \approx 0$.  
The time required for $R$ to increase by one unit corresponds to the time individuals remain infected.  
Given that  
$$
R(t) = \gamma \int_{0}^t I(s) ds + R(0)
$$
considering  $R(0)=0$, for small $R$, and $I$ , we obtain  
$$
R(t) \approx \frac{\gamma I_0}{\beta - \gamma} \left( e^{(\beta -\gamma) t} -1\right) \quad \text{for small times} \quad R(t) \approx \gamma I_0 t
$$
Thus, the time required for $R$ to increase by one unit is $1/\gamma$. Therefore,  
$$
\frac{1}{\gamma} = \text{~~ the time an individual remains infected (or transmits the disease).}
$$
As observed, the students remained infectious for about 2 days before recovering. Hence, we adopt the estimate  
$$
\frac{1}{\gamma} = 2
$$

{\it Estimating $r_0$}.  
We already know that $1/\gamma \approx 2$, and therefore, we need to estimate $\beta$.  
First, we obtain a rough estimate for $r_0$. Using $S \approx 1$, $I \approx 0$, and $R = 0$,  
we can approximate the growth of $I$ using the linear term:  
$$
I(t) \approx I_0 e^{(\beta - \gamma)t} \quad \Rightarrow \quad  
\ln \frac{I(4)}{I(3)} = \ln 3 \approx (\beta - \gamma)  \quad \Rightarrow \quad \beta_{\text{initial}} \approx 1.6 \Rightarrow r_0 \approx 3.2
$$  

However, this approximation is not very accurate as we start at day 3, when the number of infected individuals is not negligible.  


Substituting the value of $r_0$ into the formula for $I_{\text{max}}$, we obtain  
$$
I_{\text{max}} \approx 0.33 \quad \Rightarrow \quad I_{\text{max}} \times 763 \approx 247.
$$

~

\noindent
{\it Method 2: Estimating peak from the invariant.} We know that a single boy was ill on day one and that 19 boys remained susceptible. This provides
$$
r_0 = \frac{\ln\left( \frac{S_0}{S_{\infty}}\right)}{S_0 - S_{\infty}} \approx 3.8
$$
this yields
$$
I_{\text{max}} \approx 0.38 \quad \Rightarrow \quad I_{\text{max}} \times 763 \approx 293.
$$
This predicts the peak with good precision. In fact, we can use these parameters to predict the spread behavior over time.  
\vspace{-1cm}

\begin{figure}
\centerline{\hbox{\psfig{file=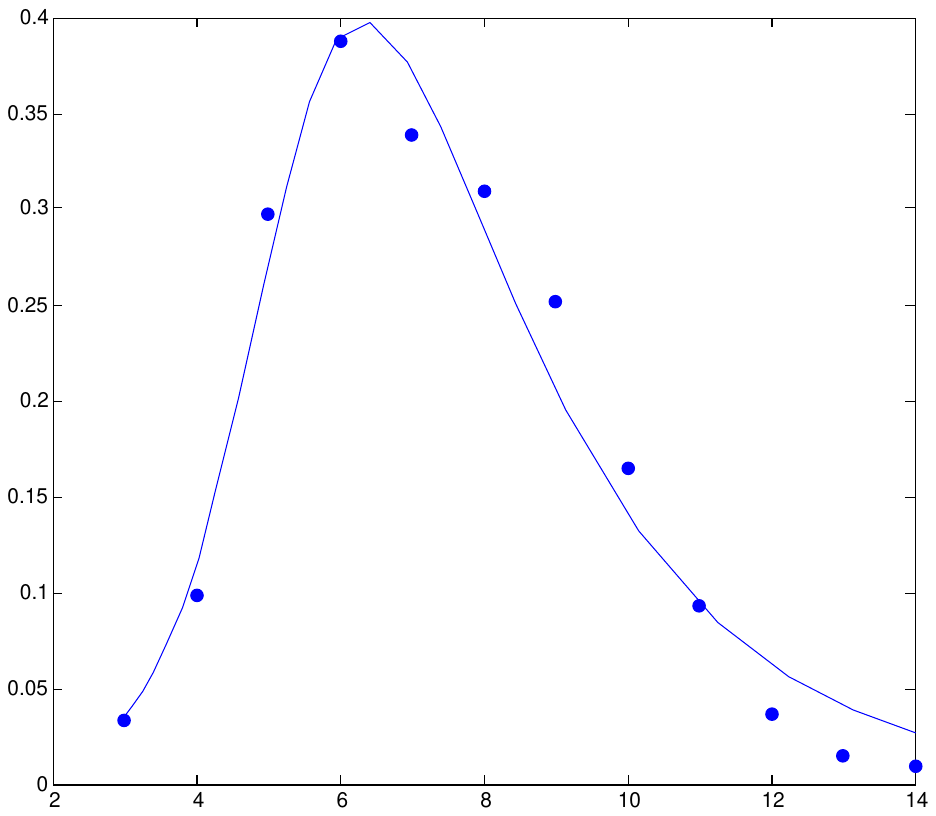,width=6.0cm}}}
\vspace{-3cm}
\caption{$I$ (dots) e prediction (full lines) as a function of time. }
\label{GraphEx}
\end{figure}

\noindent
{\it Method 3: Optimization.} One of the problems we face in fitting these models to data is that we don't really observe all states of the model. That is, as we have seen above, the data given is usually the time series of the infected individuals \(I_{\text{data}}(t_i) \), where \(t_i \) is the observed number of infected individuals at times \(t_i\).  This data is also often corrupted by delays in reporting cases, misreporting, such as underestimation, and incorrect data processing. We will assume that none of these issues are present in the infected data right now. That is, the number of infected is faithful. 

Even in this case, the states $S$ and $R$ are not provided. One way around it is to run a computational optimization method to minimize
\[
L(\beta, \gamma, S_0, I_0, R_0) = \min_{\beta, \gamma} \, \sum_{i=1}^N \left( I_{\text{model}}(t_i; \beta, \gamma, S_0, I_0, R_0) - I_{\text{data}}(t_i) \right)^2
\]

where:
\(I_{\text{model}}(t_i; \beta, \gamma, S_0, I_0, R_0)\) is the predicted number of infected individuals at time \(t_i\), calculated using the SIR model and . The unknowns are then the initial values \(S_0, I_0, R_0\) and the model parameters   \(\beta\) and \(\gamma\) to be optimized.

The goal is to find the values of \(\beta\) and \(\gamma\) that minimize the sum of the squared differences between the predicted and observed data points. The bounds for \(\beta\) and \(\gamma\) are restricted to the range \([0, 1]\) to ensure meaningful and feasible parameter values.

\section*{Optimization Process}

We provide a code for the optimization adapted from a code generated by ChatGPT.  
The optimizer used in this code, \texttt{scipy.optimize.minimize}, works by iteratively adjusting the parameters of the model (\(\beta\) and \(\gamma\)) to minimize the loss function. The optimizer works as follows: 

1. \textbf{Initialization:} The optimizer starts with an initial guess for the parameters (\(\beta\) and \(\gamma\)). These values are provided as the \texttt{initial\_guess} in the code.

2. \textbf{Evaluation:} The optimizer evaluates the loss function:
\[
L(\beta, \gamma, S_0, I_0, R_0) = \sum_{i=1}^N \left( I_{\text{model}}(t_i; \beta, \gamma, S_0, I_0, R_0) - I_{\text{data}}(t_i) \right)^2
\]
It computes the difference between the predicted infected values (\(I_{\text{model}}\)) and the observed data (\(I_{\text{data}}\)) at all time points and sums the squared differences.

3. \textbf{Parameter Update:} The optimizer adjusts \(\beta\) and \(\gamma\) slightly in the direction that reduces the loss. It respects the bounds (\(0 \leq \beta, \gamma \leq 1\)) during this adjustment.

5. \textbf{Convergence Check:} After each iteration, the optimizer checks if the loss has decreased sufficiently or if the parameter updates are small enough to stop.

~ \\

{\it Results.} Using the algorithm, we got $\beta=1.7$ and  $\gamma = 0.45$ providing $r_0 = 3.77$.

\section{Coronavirus in Brazil}

The data for the number of infected individuals with coronavirus in Brazil are:

\begin{center}
\begin{tabular}{cc}
\hline
Day & \# Infected\\
\hline
\hline
29/02/20	& 2\\
04/03/20	&4\\
05/03/20	&8\\
06/03/20	&14\\
07/03/20	&19\\
08/03/20	&25\\
09/03/20	&30\\
10/03/20	&34\\
11/03/20	&52\\
12/03/20	&76\\
13/03/20	&98\\
14/03/20	&121\\
15/03/20	&200\\
16/03/20	&234\\
17/03/20	&291\\
18/03/20	&428\\
19/03/20	&621\\
20/03/20	&964\\
21/03/20	&1,178\\
22/03/20	&1,546\\
\hline
\end{tabular}
\end{center}

{\bf Exercise} Calculate $r_0$ starting from the $100^\text{th}$ infected individual. \\  

{\bf Exercise} Assume that $\gamma = 1/5$. Determine $\beta$. \\  

{\bf Exercise} Assume that $I$ is given by the Bernoulli equation. Calculate the number of infected individuals on May 1st.

\chapter{Networks\\ {\small Lecture by Zheng Bian}}

\noindent\hfill
\begin{minipage}[t]{9cm}
{\small 
We establish that the spectral radius of a graph dictates the exponential growth rate of the number of walks (Proposition 1). These results will play a major role in establishing the epidemic threshold for SIS models on networks. We then introduce a class of random graphs and discuss their spectral radius in Theorem 3.4. Readers whose interests lie outside graph-theoretic aspects may verify this result independently and proceed to the next lecture.
}
\end{minipage}

\section{Networks}
A network is a graph $G$ comprising a set of $N$ 
nodes (or vertices) connected by a set of $M$ links (or edges).  Graphs are 
the mathematical structures used to model pairwise relations between objects. 
We shall often refer to the network topology, which is the layout pattern of 
interconnections of the various elements. Topology can be considered as a 
virtual shape or structure of a network. 

The networks we consider here are  {\it simple} and {\it undirected}. 
A network is called {\it simple} if the nodes do not have self-connections, and  
{\it undirected} if there is no distinction between the two vertices associated 
with each edge. A {\it path} in a graph is a sequence of connected (non-repeated) 
nodes.  From each node of a path, there is a link to the next node in the sequence.  
The length of a path is the number of links in the path. See further details in Ref. \cite{CL06}. 
We consider unweighted graphs. Weighted graphs may have more interesting behaviour such as genericity \cite{poignard2018spectra}, but we will leave it out for sake of simplicity.

For example, let's consider the network in Fig. \ref{GraphEx}a). Between the nodes
$2$ and $4$ we have three paths $\{2,1,3,4\}$, $\{2,5,3,4\}$ and $\{2,3,4\}$. The first 
two have length $3$, and the last has length $2$. Therefore, the path $\{2,3,4\}$ is the 
shortest path between the node $2$ and $4$. 

The network {\it diameter} $d$ is the greatest length of the shortest path between any pair of 
vertices.  To find the diameter of a graph, first find the shortest path between 
each pair of vertices. The greatest length of any of these paths is the diameter 
of the graph. If we have an isolated node, that is, a node without any connections, 
then we say that the diameter is infinite. A network of finite diameter is called {\it connected}. 

A connected component of an undirected graph is a subgraph with a finite diameter. 
The graph is called {\it directed} if it is not undirected. If the graph is directed,
then there are two connected nodes, say, $u$ and $v$, such that $u$ is reachable from $v$, 
but $v$ is not reachable from $u$.  See Fig. \ref{GraphEx} for an illustration. 

The network may be described in terms of its {\it adjacency matrix}  $\bm{A}$, which 
encodes the topological information, and is defined as 
$$
A_{ij} = \left\{
\begin{array}{cc}
1 & \mbox{if } i \mbox{ and } j  \mbox{ are connected} \\
0 & \mbox{ otherwise }.
\end{array}
\right.
$$

An undirected graph has a symmetric adjacency matrix. The {\it degree} $k_i$ of the 
$i$th node  is the number of connections it receives, clearly 
$$
k_i = \sum_{j=1}^n A_{ij}.
$$ 

The networks we encounter in real applications have a wilder
connection structure. Typical examples are cortical networks, 
the Internet, power grids, and metabolic networks \cite{Newman}. These networks don't have a regular structure of connections. We will discuss an important class of networks soon, namely the random graph model with an expected degree, which allows for multiple mathematical insights.

\section{Spectral radius of a network}

Since the adjacency matrix $A$ of an undirected graph $G$ is symmetric, it follows by the Spectral Theorem that $A$ admits an orthonormal eigenbasis $\{\phi_1,\cdots,\phi_n\}$, where $\phi_i$ corresponds to the eigenvalue $\mu_i$. Moreover, the eigenvalues $\mu_i$ are real and we index them so that $|\mu_1|\geq\cdots\geq |\mu_n|$. By the spectral radius, we mean $\rho(A)=|\mu_1|$. We will also speak about the spectral radius of a graph $G$, and we mean the spectral radius of the corresponding adjacency matrix. 

The spectral radius and the operator norm are related by Gelfand's formula
\[
\rho(A) \leq  \|A^k\|^{1/k} \text{ for all $k=1,2,\cdots$, and }\rho(A)= \lim_{k\to\infty} \|A^k\|^{1/k}.
\]
Now, since $A$ is a symmetric matrix, we obtain $\rho(A)=\|A\|_2$, where $\|\cdot\|_2$ denotes the Euclidean norm, also called the spectral norm.
This number $\rho(A)$ gives the exponential rate of growth for the number of walks: 
$$\#\{\text{walks in $G$ of length $k$}\} = O(\|A^k\|)=O((\rho(A))^k)$$


\begin{prop}[Number of walks] \label{prop:walks_in_G}
The number of walks of length $k$ in $G$ starting from $i$ and ending at $j$ is given by $(A^k)_{ij}$. The total number of walks of length $k$ is $$N_k=\langle \One, A^k \One\rangle$$ If $G$ is connected, then the exponential growth rate of $N_k$ is given by 
$$\rho(G)=\lim_{k\to+\infty} (N_k)^{1/k}$$
\end{prop}

\begin{proof}
For the first assertion, we induct on $k$. The base case $k=1$ is clear from the definition of $A$. Now assume $k$ and show $k+1$. 
\[
(A^{k+1})_{ij} = (A A^{k})_{ij} = \sum_{\ell=1}^n A_{i\ell} (A^k)_{\ell j}.
\]
By induction hypothesis $k$, the number of walks of length $k$ in $G$ starting from $\ell$ and ending at $j$ is given by $(A^k)_{\ell j}$. Any walk of length $k+1$ in $G$ starting from $i$ and ending at $j$ can be uniquely decomposed into two steps: the first of length $1$ from $i$ to $\ell$ and the second of length $k$ from $\ell$ to $j$. This proves the first assertion.

For the second assertion, we compute
\begin{align*}
	\langle \One, A^k\One\rangle = \sum_j 1\cdot (A^k\One)_j = \sum_j \sum_i (A^k)_{ij}\cdot 1.
\end{align*}

For the last assertion, we note that $A$ is an irreducible (because $G$ is connected) matrix with nonnegative entries; it follows from Perron-Frobenius Theorem that $\mu_1=\rho(A)$ is a simple eigenvalue of $A$ and $\phi_1$ has positive entries. 

Write $\One=\sum_i \alpha_i \phi_i$ as a linear combination of the orthonormal eigenbasis $\{\phi_i\}$ of $A$ and we have $\alpha_1\neq 0$. Now
\begin{align*}
	(N_k)^{1/k} =& \langle \One, A^k\One\rangle^{1/k} = \left\langle \sum_i \alpha_i \phi_i, A^k \sum_j \alpha_j \phi_j \right\rangle^{1/k} \\
	=& \left(\sum_i |\alpha_i|^2 \mu_i^k\right)^{1/k}
\end{align*}
Let $\rho = |\lambda_1| = |\lambda_2|=\cdots = |\lambda_m| > |\lambda_{m+1}| \ge |\lambda_{n}|$. That is, we are assuming that the multiplicity of the spectral radius is $m\le n$. This implies that
$$
\sum_i |\alpha_i|^2 \mu_i^k = \rho^k p \left(  1 +\frac{1}{p}\sum_{i=m+1}^n |\alpha_i|^2 r_i^k \right) 
$$
where $p=\sum_{i=1}^m |\alpha_i|^2$, and $r_i = \lambda_{i}/\rho<1$ for $i>m+1$, thus, $r_i^k\rightarrow 0$ as $k\rightarrow \infty$. Moreover, given a fixed number $p>0$ we have $\lim_{k \rightarrow \infty}p^{1/k}=1$. Bringing all this together, we have
\begin{align*}
	\lim_{k\rightarrow \infty }(N_k)^{1/k} =& \lim_{k \rightarrow \infty} \left( \rho^k p \left(  1 +\frac{1}{p}\sum_{i=m+1}^n |\alpha_i|^2 r_i^k \right) \right)^{1/k} \\
	=& \, \rho 
\end{align*}

\end{proof}

\begin{Obs}
	If the nodes of a network $G$ represent individuals or populations of individuals and the edges contact patterns among them, then a walk in $G$ can be interpreted as the spread of an infectious disease. The more possible walks in $G$, the more routes a disease can spread in $G$. In light of Proposition \ref{prop:walks_in_G}, $\rho(G)$ serves as a good proxy for the capacity of $G$ to spread diseases.
\end{Obs}

We compute the spectral radius for two extreme cases. The star $S_m$ with $m$ leaves consists of one hub together with $m$ low-degree nodes, where the edges exist only between the hub and each low-degree node. A graph $G$ is said to be $k$-regular if each node in $G$ has degree equal to $k$.
\begin{prop}
The star graph $S_m$ has spectral radius $\rho(S_m)=\sqrt{m}$ and a connected $k$-regular graph $G$ has spectral radius $\rho(G)=k$.
\end{prop}

\begin{proof}
	For the star, we set up the eigenequation $Av=\lambda v$ that leads to
	\begin{align*}
		v_2+\cdots+v_{m+1} =& \, \lambda v_1\\
		v_1 =& \, \lambda v_i,~~~~i=2,\cdots,m+1.
	\end{align*}
It follows that either $\lambda=0$ or $m v_1/\lambda=\lambda v_1$ with $v_1=\lambda v_i\neq 0$. In the second case, we obtain $\lambda = \pm\sqrt{m}$. We conclude $\rho(S_m)=\sqrt{m}$.
	
	For the connected $k$-regular graph $G$, observe that $A(G)$ is irreducible with nonnegative entries and row sums all equal to $k$. It follows from Perron-Frobenius theorem that $\rho(G)=k$ is the dominant simple eigenvalue.
\end{proof}

The spectral radius of a complex network is generally difficult to compute, but some useful estimates may be obtained. We first present a handy lemma to upper bound the spectral radius of a subgraph.

\begin{lemma}[Subgraph lemma]\label{lemma:subgraph_rho}
If $G'$ is a subgraph of $G$, then $\rho(G')\leq \rho(G)$.
\end{lemma}

To prove Lemma \ref{lemma:subgraph_rho}, we will make use of the following two lemmas.

\begin{lemma}[Perron-Frobenius; \cite{Horn2012} Theorem 8.3.1] \label{lemma:PF_nonneg}
	If $A$ is a (possibly reducible) matrix with nonnegative entries, then $\rho(A)$ is an eigenvalue of $A$, and there is $x\neq 0$ with nonnegative entries such that $Ax=\rho(A)x$.
\end{lemma}

\begin{lemma}[Neighborless node]\label{lemma:neighborless_node}
If $G$ has a neighborless node $i$, then the subgraph $G'$ obtained by removing $i$ from $G$ has the same spectral radius $\rho(G')=\rho(G)$.
\end{lemma}
\begin{proof}
	With $i=1$, the characteristic equation for $A(G)$ reads 
		\[
	\det\begin{pmatrix}
		-\lambda & \mathbf{0}^{\top}\\
		\mathbf{0}&A(G')-\lambda I
	\end{pmatrix} = -\lambda \det (A-\lambda I).
	\]
	This shows that the eigenvalues of $A(G)$ have an extra 0 compared to $A(G')$. 
\end{proof}

\noindent \textit{Proof of Lemma \ref{lemma:subgraph_rho}.} In light of Lemma \ref{lemma:neighborless_node}, we may assume that $G'$ is obtained from $G$ by removing only edges (not nodes). By Lemma \ref{lemma:PF_nonneg}, there is $y$ with nonnegative entries and with $\|y\|=1$ such that
\begin{align*}
	\rho(G') =& \rho(G') \langle y,y\rangle = \langle y, A' y\rangle \\
	=& \langle y, (A+A'-A)y\rangle = \langle y, Ay\rangle + \langle y, (A'-A)y\rangle \\
	\leq& \langle y, Ay\rangle \leq \sup_{\|x\|=1} \langle x, Ax\rangle = \rho(G).
\end{align*}
The first inequality follows from the fact that $A'-A$ is nonpositive and $y$ has nonnegative entries; the last equality uses the fact that $A$ is a symmetric matrix. 
\qed

Combining the star graph spectral radius in Proposition \ref{prop:walks_in_G} and the subgraph lemma \ref{lemma:subgraph_rho}, we obtain the following.
\begin{theorem}[$\sqrt{m}$ lower bound of graph spectral radius] \label{thm:spectra_radius_lower_bound_G(w)_sqrt(m)}
	If $G$ has a node of degree $m$, then $\rho(G)\geq \sqrt{m}$.
\end{theorem}

\section{Random graphs from a given expected degree sequence}

The random graph model with expected degree sequence, introduced by Fan Chung, generalizes the classical Erd\"os-R\'enyi random graph by incorporating a prescribed degree sequence. This model captures the heterogeneous degree distributions observed in real-world networks, such as power-law degree distributions.


{\it Model Description:} 
Given a sequence \(\mathbf{w}=(w_1, w_2, \ldots, w_n)\) of positive numbers $w_i$, the graph model $G(\mathbf{w})$ characterizes the links between any two nodes $i$ and $j$ of the graph as independent Bernoulli random variables with success probability 
\[
p_{ij} = \frac{w_i w_j}{\sum_{k=1}^n w_k}
 \mbox{~ with ~ }m = \max w_i \mbox{~ and  impose ~} 
m^2 \le {\sum_{k=1}^n w_k}
\]
to ensure that $p_{ij}$ is a probability.  The expected degree $k_i$ is
\[
\mathbb{E}[k_i] = \sum_{j} p_{ij} = \sum_{j} \frac{w_i w_j}{\sum_{k=1}^n w_k} = w_i
\]
\begin{Obs}
If one excludes self-loops (e.g., simple graphs), then $p_{ii}=0$ and $\mathbb{E}[k_i] = w_i \frac{-w_i+\sum_{k}w_k}{\sum_k w_k}= w_i\left[1-\frac{w_i}{\sum_k w_k}\right]$. In practical scenarios, this discrepancy does not make a difference.
Note that any graph can occur in $G(\mathbf{w})$, maybe with a small probability.
\end{Obs}

The degrees of this model are jointly concentrated in the expected sequence, as stated in the following

\begin{lemma}[Joint concentration of actual degrees; \cite{CL06} Lemma 5.7] \label{thm:CLlemma5.7}
	\index{concentration inequality of actual degrees in Chung-Lu model}
	For a graph $G$ in $G(\brw)$, with probability $1-O(n^{-1/5})$ all vertices satisfy
	$$|k_{i}- w_i| \leq 2 (\sqrt{w_i \log n} + \log n).$$
\end{lemma}

\begin{proof}
The proof is presented in the appendix \ref{CIn}
\end{proof}

\section{Spectral radius upper bound for $G(\brw)$}

The Chung-Lu model helps analyze spectral properties of networks, which are crucial for applications such as community detection, robustness analysis, and network optimization. The random graph model with an expected degree sequence, proposed by Fan Chung, is a versatile framework for generating graphs that reflect the heterogeneous degree distributions observed in real-world networks. Its flexibility and ability to model power-law degree distributions make it a powerful tool for studying network structure and dynamics.
Recall the Chung-Lu random graph model with expected degree sequence $\brw=(w_1,\cdots,w_n)$, with 
$$
\tilde{w}= \frac{\sum_{i=1}^n w_i^2}{\sum_{i=1}^n w_i} = \frac{\langle w^2 \rangle }{\langle w \rangle }
$$ 
Both $m$ and $\tilde{w}$ are assumed to grow faster than $\log n$. We will prove the following bound for the spectral radius of $G$ from $G(\brw)$.

\begin{theorem}[Spectral radius upper bound for $G(\brw)$; \cite{CL06} Lemma 8.8 and Theorem 8.9] \label{thm:spectra_radius_upper_bound_G(w)}
The spectral radius $\rho(G)$ for a random graph $G$ from $G(\brw)$ satisfies the following upper bound with predominant probability (i.e., the exceptional probability tends to 0 as the size $n\to\infty$).
\begin{equation}\label{eq:spectral_radius_upper_bound_G(w)}
	\rho(G) \leq \tilde{w} + \sqrt{
	6\sqrt{m\log n} (\tilde{w}+\log n)
	}
	+ 3 \sqrt{m \log n} \leq 7 \sqrt{\log n} \max\{\sqrt{m},\tilde{w}\}.
\end{equation}
\end{theorem}

We will make use of the following lemma.
	\begin{lemma}[Perron's lemma; \cite{CL06} Lemma 8.4]
		If $A$ is an $n\times n$ symmetric matrix with nonnegative entries and $c_1,\cdots,c_n>0$, then
		\[
		\rho(A)=\|A\|_2 \leq \,\, \max_{i=1,\cdots,n} \,\,  \frac{1}{c_i} \sum_{j=1}^n c_j A_{ij}.
		\]
	\end{lemma}
	\begin{proof}
		Let $C:= \mathrm{diag}(c_1,\cdots,c_n)$. Then, $A$ and $C^{-1}AC$ have the same eigenvalues and $C^{-1}AC$ also has nonnegative entries. By Gelfand's formula, we have
		$
		\rho(A) = \rho(C^{-1}AC) \leq \|C^{-1}AC\|_{\infty} = \max_i \sum_j (C^{-1}AC)_{ij} = \max_i \sum_j \frac{1}{c_i} (AC)_{ij} = \max_i \frac{1}{c_i} \sum_j A_{ij} c_j.
		$
	\end{proof}

\noindent \textit{Proof of Theorem \ref{thm:spectra_radius_upper_bound_G(w)}.}
Define constants
\[
c_i := \begin{cases}
	w_i,&\text{if }w_i>x\\
	x,&\text{if }w_i\leq x
\end{cases},
\]
where $x>0$ is a constant to be chosen later. Define random variables
\[
X_i := \frac{1}{c_i} \sum_{j=1}^n c_j A_{ij},
\]
where $A_{ij}$ are the Bernoulli random variables from $G(\brw)$ whose success corresponds to an edge between $i,j$. In particular, we have
\begin{align*}
	\mathbb{E}[X_i] =& \frac{1}{c_i} \sum_j c_j \mathbb{E}[A_{ij}] = \frac{1}{c_i} \sum_j c_j  w_iw_j /\sum_k w_k\\
	=& \begin{cases}
		\sum_{w_j>x} w_j^2 / \sum_k w_k + x \sum_{w_j\leq x} w_j/\sum_k w_k,& \text{if }w_i>x\\
		\frac{w_i}{x} \sum_{w_j>x} w_j^2 /\sum_k w_k + w_i \sum_{w_j\leq x} w_j / \sum_k w_k,& \text{if }w_i\leq x
	\end{cases}\\
	\leq& \tilde{w}+x.
\end{align*}
And
\begin{align*}
\mathrm{Var}(X_i) =& \mathrm{Var}\left(\frac{1}{c_i} \sum_j c_j A_{ij}\right) = \frac{1}{c_i^2} \sum_j c_j^2 \mathrm{Var}(A_{ij}).
\end{align*}
With $\mathrm{Var}(A_{ij}) = \mathbb{E}[A_{ij}^2] - \mathbb{E}[A_{ij}]^2 = \mathbb{E}[A_{ij}] - \mathbb{E}[A_{ij}]^2 = w_iw_j/\sum_k w_k - \left(w_iw_j/\sum_k w_k\right)^2 \leq w_iw_j / \sum_k w_k$, we continue
\begin{align*}
	\mathrm{Var}(X_i) \leq &\frac{1}{c_i^2} \sum_j c_j^2 w_iw_j/\sum_k w_k\\
	\leq& \begin{cases}
		\frac{1}{w_i} \sum_{w_j>x} w_j^3 /\sum_k w_k + \frac{x^2}{w_i} \sum_{w_j\leq x} w_j /\sum_k w_k,& \text{if }w_i>x\\
		\frac{w_i}{x^2} \sum_{w_j>x} w_j^3 /\sum_k w_k + w_i \sum_{w_j<x} w_j /\sum_k w_k,&\text{if }w_i\leq x
	\end{cases}\\
	\leq& \frac{m}{x} \tilde{w} + x.
\end{align*}
From a concentration result similar to Chernoff Theorem \ref{thm:exp_concentration+}, it follows that 
\[
\PP(|X_i-\mathbb{E}[X_i]|\geq a) \leq \exp\left(- \frac{a^2}{2(\mathrm{Var}(X_i)+ ma/3x)}\right).
\]
In particular, for each fixed $i$, we have
\begin{equation}\label{eq:spectral_radius_upper_bound_eq_X_i}
X_i \leq \mathbb{E}[X_i] + a \leq \tilde{w}+x+a
\end{equation}
with exceptional probability $\leq \exp\left(- \frac{a^2}{2(\mathrm{Var}(X_i)+ ma/3x)}\right) \leq \exp\left(- \frac{a^2}{2(\frac{m}{x}\tilde{w}+x+ ma/3x)}\right)$.
Choosing
$
x=\sqrt{m\log n},~a=\sqrt{
	6\left(\frac{m}{x}\tilde{w}+x\right) \log n
} + \frac{2m}{x} \log n,
$
we upper bound this by 
\begin{align*}
\exp\left(- \frac{a^2}{2(\frac{m}{x}\tilde{w}+x+ ma/3x)}\right) = \exp\left(-\log n \cdot o(1)\right) = o(1/n).
\end{align*}
By removing the exceptional probability for each $i=1,\cdots,n$, we obtain
\begin{align*}
\max_iX_i\leq& \tilde{w}+x+a = \tilde{w}+ \sqrt{m\log n } + \sqrt{
	6\left(\frac{m}{x}\tilde{w}+x\right) \log n
} + \frac{2m}{x} \log n\\
\leq &\tilde{w} + \sqrt{
	6\sqrt{m\log n} (\tilde{w}+\log n)
}
+ 3 \sqrt{m \log n} \leq 7 \sqrt{\log n} \max\{\sqrt{m},\tilde{w}\}.
\end{align*}
The proof is complete by applying Perron's Lemma.
\qed
 
 \section{$\tilde{w}$ spectral radius lower bound for $G(\brw)$}
 
 
 \begin{theorem}[\cite{CL06} Lemma 8.7] \label{thm:spectra_radius_lower_bound_G(w)_tilde(w)}
 	The spectral radius $\rho(G)$ of a random graph $G$ from $G(\brw)$ satisfies the following lower bound with predominant probability
 	\[
 	\rho(G) \geq \tilde{w}-o(1).
 	\]
 \end{theorem}
 

 \begin{theorem}[\cite{CL06} Theorems 8.10, 8.11]
With predominant probability, the spectral radius $\rho(G)$ of a random graph $G$ from $G(\brw)$ is 
\begin{itemize}
\item[(i)] \, $\rho(G) = (1+o(1))\tilde{w}$ if $\tilde{w}>\sqrt{m}\log n$ 
\item[(ii)] \, $\rho(G) = (1+o(1))\sqrt{m}$ if $\sqrt{m}>\tilde{w}\log^2 n$
\end{itemize}
 \end{theorem}
 
 \begin{proof}
  Combining Theorems \ref{thm:spectra_radius_upper_bound_G(w)}, \ref{thm:spectra_radius_lower_bound_G(w)_tilde(w)}, we obtain (i). For (ii), Theorem   \ref{thm:spectra_radius_lower_bound_G(w)_sqrt(m)} gives lower bound $\sqrt{m}$, and Theorem \ref{thm:spectra_radius_upper_bound_G(w)} can give $\rho(G)\leq \sqrt{m}(3\sqrt{\log n}+o(1))$. To obtain the sharper bound $\sqrt{m}(1+o(1))$ requires more work, see \cite[Section 8.5.1]{CL06}.
 \end{proof}

Lower eigenvalues of $A$ can also be meaningful, and there are estimates for them, see \cite{CL06} for details.

\section{Examples of Random Graphs}

{\it Homogeneous random graphs $G_{\rm homo}$}. Consider the case  where
$$w_i = n p \mbox{~ for all ~} i = 1,\dots,n$$
This case is called the Erdos-Renyi random graph, and it is shown 
to have a giant component when $p > O(\log n /n)$ \cite{CL06}. Lets consider 
$p> O(\log^k n /n)$. In this case, 
$$
\tilde w =  O(\log^k n) > \sqrt{m} \log n = O( \log^{k/2}n) \log n = O(  \log^{(k+1)/2} n) 
$$
$$
(k+1)/2 < k  \Rightarrow 2k > k+1 \Rightarrow k> 1
$$
which implies that 
$$
\rho(G_{\rm homo}) = np (1+o(1)) 
$$

\noindent
{\it Heterogeneous Random Graphs $G_{hub}$}. Let's consider the following sequence
$$
\bm{w} = (m, \underbrace{np,\cdots, np}_{n-1 \mbox{~times}})
$$
again  $p=O((\log^2 n)/n)$. Next, we take 
$$
m = A^2 \sqrt{n} \,\, \Rightarrow \,\, A^4 n \le \sum_{i=2}^n np = O(n \log^k n)
$$  
thus, the graphical condition is satisfied, and 
\begin{eqnarray}
\tilde w &=& \frac{m^2 + \sum_{i=2}^n w_i^2}{m + \sum_{i=2}^n w_i} \nonumber \\
&=& \log^k n (1+ o(1))\nonumber
\end{eqnarray}
that is, this node in the model that plays the role of a hub. 
Now, clearly
$$
\sqrt{m} = O(n^{1/4}) > \tilde w \log^2 n = O(\log^{k+2} n) \,\,\, \Rightarrow \,\,\,\,
\rho(G_{\rm hub}) = A \sqrt{m}
$$

 \section{Spectral radius of random power-law graphs}
 
{\it Heuristics}. Let's discuss some cases to gain an insight into the results. 
For instance, in classical power law graphs 
we have 
$
P(k) \propto k^{\beta}.
$ Let $w$ be the expected average degree.

\begin{theorem}[Spectral radius of random power-law graphs]\label{PLRG}
\begin{enumerate}
	\item For $\beta \geq 3$, suppose the maximum degree $m$ satisfies
	\[
	m > w^2 \log^3 n
	\]
 Then almost surely $\rho(G) = (1 + o(1))\sqrt{m}$.
	
	\item For $2.5 < \beta < 3$, suppose $m$ satisfies
	\[
	m > w^{\frac{\beta - 2}{\beta - 2.5}} \log^{\frac{3}{\beta - 2.5}} n.
	\]
	Then almost surely $\rho(G) = (1 + o(1))\sqrt{m}$.
	
	\item For $2 < \beta < 2.5$ and $m > \log^{\frac{3}{2.5 - \beta}} n$, almost surely  $\rho(G) = (1 + o(1)) \tilde{w}$.
\end{enumerate}
\end{theorem}

\begin{proof}
{
The proof is accessible by combining general spectral radius bounds for $G(\brw)$ and the relation between $\sqrt{m},\tilde{w}$ in the power-law case.\\
}
\end{proof}

\chapter{Lecture 5: Multiple Compartments}

\noindent\hfill
\begin{minipage}[t]{9cm}
{\small
We introduce the SIS model on a network and derive a condition for the exponential stability of the disease-free equilibrium. Using this condition, we show that networks with a power-law degree distribution exhibit a zero epidemic threshold in the large-graph limit. }
\end{minipage}

~\\
~\\

In many cases, it is of interest to split the population into multiple subgroups.
This might happen, for instance, when certain age groups are more prompt to get infected or some age groups show a stronger response to the infection utilizing more intensively the health care systems.  Early pandemic models demonstrated that regions with greater intergenerational contact had faster virus transmission to high-risk groups such as elders \cite{prem2020age,coelho2022prevalence,coelho2024sars}.

Indeed, at the beginning of the COVID-19 pandemic, Italy suffered significantly more than Germany in terms of cases and fatalities. One hypothesis for this disparity involves differences in the interactions across age groups, particularly the frequency of interactions between children and elders.  Italy is characterized by a high proportion of multigenerational households as well as frequent interactions between elders and their children or grandchildren. Moreover, cultural norms that encourage close family ties, including caregiving by younger family members. This means that if a new infection starts, say, at a school because of the elderly-youth interaction, it will spread across the elderly population quickly. 

In contrast, Germany has a lower proportion of multigenerational households as well as a higher likelihood of elderly individuals living independently or in care facilities.
The POLYMOD study provided contact matrix data for European countries \cite{mossong2008social}. It showed that Italy had significantly higher interaction rates between children and older adults compared to Germany. These differences likely contributed to a higher rate of transmission to elders in Italy.

Germany's early emphasis on protecting eldercare facilities and limiting intergenerational contact likely contributed to its lower initial fatality rates. Meanwhile, Italy's healthcare system, especially in Lombardy, was overwhelmed early in the pandemic due to the rapid spread among the elderly.

The hypothesis that Italy suffered more than Germany in the early stages of COVID-19 due to higher intergenerational contact is supported by demographic, cultural, and epidemiological data. Differences in contact matrices, healthcare preparedness, and policy responses likely explain much of the disparity between the two countries \cite{dowd2020demographic}.

Another scenario where splitting the population into subgroups is useful is in a country 
where an infection starts in a given city and spreads across the country. In this case, it is desirable to understand the impact of commute and better plan the local health care system requirements. We will discuss these two cases.

A city-stratified compartmental model is crucial for understanding and managing the spatial spread of infectious diseases like COVID-19. By dividing the population into cities, the model captures the unique dynamics of each locality, such as differences in population density, healthcare infrastructure, and mobility patterns. This granularity allows for a deeper understanding of how diseases propagate between and within cities, especially when combined with geolocation-based mobility data. Transition matrices derived from these data provide real-world insights into movement patterns, enabling the identification of high-risk areas and predicting how interventions like lockdowns or travel restrictions affect disease spread.

This approach also helps optimize resource allocation by prioritizing regions likely to experience outbreaks, including those that are geographically distant but economically connected to infection hubs. Stratified models are adaptable and can be extended to include more compartments or phases of the disease, making them suitable for both short-term containment and long-term planning. Their ability to localize predictions ensures that policies and interventions are tailored to the specific needs of each city, improving the overall efficacy of public health measures. By balancing detail with practicality, city-stratified models provide actionable insights that are critical for managing epidemics effectively.

\section{SIS model on a network}

The SIS model on a network extends the basic SIS framework to account for the structure of interactions between individuals.
Here, each node \(i\) represents a node in the network. It could be representing a city or a computer and the like. Edges between nodes represent potential transmission paths. At any time \(t\), a individuals can be in one of two states: Susceptible or infected: Currently infected and capable of transmitting the disease. The total population at each node satisfies:
\[
S_i(t) + I_i(t) = 1, \,\,\, \mbox{~for~} i = 1, 2, \cdots, n
\]

{\it Infection Dynamics:}
The fraction of infected nodes \(I_i\) changes over time as:
\[
\frac{dI_i}{dt} = \beta S_i \left( \sum_{j=1}^N a_{ij} I_j \right) - \gamma I_i,
\]
where the adjacency matrix element (\(a_{ij} = 1\) if nodes \(i\) and \(j\) are connected, otherwise \(0\)). Moreover, using \(S_i(t) + I_i(t) = 1\), we obtain the infection dynamics:
\[
\frac{dI_i}{dt} = \beta (1 - I_i) \left( \sum_{j=1}^N a_{ij} I_j \right) - \gamma I_i.
\]

\section{Linear Stability Analysis of the disease-free equilibrium}

Consider the following equilibrium points
$$
(S_i, I_i, R_i) = (1,0,0) \,\,\,\, \mbox{~for all~} i = 1,\cdots, n
$$
Lets denote 
$$
\bm{S} = (S_1,\cdots,S_n)^*,\,\,\,  \bm{I} = (I_1,\cdots,I_n)^* \mbox{~and~}\bm{R} = (R_1,\cdots,R_n)^*
$$
here, $^*$ denotes the transpose.
Using this notation and linearizing the equation around the disease-free equilibrium, we obtain the 
linear equation
$$
\frac{d \bm{I} }{dt} = \frac{1}{\gamma} \left( r_0 \bm{A}    -  \mathbb{1} \right) \bm{I}
$$
From now on $\mathbb{1} $ will denote the identity matrix. 
This system will be stable when 
$$
\rho(r_0 \bm{A} - \mathbb{1}) = r_0 \rho(\bm{A}) -1 <0 
$$
Here $R_0$ the effective reproduction number is defined as 
$$
R_0 = r_0 \rho(\bm{A})
$$
so it depends on the biological part as well as the interactions through the network.
As we proved in the Appendix. 
This gives us the following result
\begin{Teorema}
Consider the SIS model on a network $G$ with adjacency matrix $A$.
The disease-free equilibrium $\emph{\bm{I}}_{\rm free} = \emph{\bm{0}}$ is exponentially stable 
if and only if 
$$
R_0 < 1
$$
\end{Teorema}
\section{Four examples}

Let's consider the following results that are easy to obtain from the discussions we had in Lecture 4.

\begin{enumerate}
\item {\it Spreading in k-regular graphs does not depend on the network size.}
For a k-regular graph, as we discussed earlier, we have that 
$$
\rho(G_{k {\rm-regular}}) = k  \quad \Rightarrow \quad r_0 < \frac{1}{k}
$$
depends only on the number of contacts a given node has, and it is independent of the network size $n$. 

\item {\it It is not possible to control the spreading in a large star network.}
For a star graph of $n+1$ nodes, we have that 
$$
\rho(G_{star}) = \sqrt{n} \quad \Rightarrow \quad  r_0 < \frac{1}{\sqrt{n}} \stackrel{n\rightarrow \infty}{=}0
$$

\item {\it Random homogeneous graphs, roughly speaking, behave as homogeneous graphs.} For random homogeneous graphs, as we discussed earlier, we have that 
$np = \langle k \rangle $
$$
\rho(G_{homo}) = \langle k \rangle  \quad \Rightarrow \quad r_0 < \frac{1}{ \langle k \rangle} 
$$
in this case, $ \langle k \rangle$ is a slowly growing function of $n$. 
\item {\it A hub makes a huge impact.} Consider the graph $G_{\rm hub}$ by additing a hub in $G_{homo}$. Then, 
$$
\rho(G_{\rm hub}) \approx \sqrt{n} \quad \Rightarrow \quad r_0 < n^{-1/4} 
$$
Roughly speaking, this behaves like the star.  The function is slower decaying but still very fast. 
So the presence of a hub, or a superspreader, makes it unlikely to stop the spreading.
\end{enumerate}

\section{Examples: Internet is threshold zero in the large limit}

A random power-law graph models the internet connectivity. As stated in Theorem \ref{PLRG}
for $2 < \beta < 2.5$ and $m > \log^{\frac{3}{2.5 - \beta}} n$, almost surely we obtain
$$
\rho(G) = \tilde{w} =  \frac{\langle k^2 \rangle }{\langle k \rangle} 
$$
implying that for this class
$$
r_0 \le  \frac{\langle k \rangle }{\langle k^2 \rangle }
$$

Let's have a closer look at the power law distribution:
\[
P(k) \sim k^{-\gamma},
\]
where \(P(k)\) represents the probability of a node having degree \(k\), and \(\gamma\) is a characteristic exponent. For the internet, studies have shown that the degree distribution is well-approximated by a power law with \( \gamma \approx 2.2 \), indicating a highly heterogeneous network structure with the presence of many hubs \cite{faloutsos1999power}.

The second moment of the degree distribution is defined as:
\[
\langle k^2 \rangle = \int_{k_{{\rm min}}}^{m} k^2 P(k) \, dk = C \int_{k_{{\rm min}}}^{m} k^{2 - \gamma} \, dk,
\]
where \(k_{{\rm min}}\) is the minimum degree, \(m\) is the maximum degree in the network, \(C\) is the normalization constant. Performing the integral, we obtain \(C\) to ensure that the total probability sums to 1. When \(\gamma < 3\) we obtain
\[
\langle k^2 \rangle =  \alpha m^{3-\gamma}(1+o(1))
\]
where  $\alpha =(\gamma - 1)/[(3 - \gamma ) k_{{\rm min}}^{1-\gamma}]$.
Since \(m\) grows with the network size, \(\langle k^2 \rangle\) diverges, indicating that the network becomes increasingly dominated by high-degree nodes (hubs). 
The divergence of \(\langle k^2 \rangle\) implies that many network properties, such as the epidemic threshold, depend heavily on the network size and the largest hub. These hubs behave as super spreaders. This implies that in the large limit, we have
$$
r_0 \rightarrow 0
$$
implying that any spreading would propagate.

\section{Interpration of $\langle k^2 \rangle / \langle k \rangle$}

The average degree of the network, denoted as \( \mu \), is calculated by taking the mean of the degrees of all nodes. It is expressed as:
\[
\mu = \frac{1}{N} \sum_{i=1}^{N} d_i
\]
and tell the mean number of friends. We want to calculate
$$
{\small
\frac{\mbox{expected \# of friends of friends}}{\mbox{expected \# of friends}} = \mbox{"How many friends your friends have on average"}
}
$$

The idea is that this ratio tells whether you tend to have more friends than your friends do, if the ratio is less than one. We will show that this ratio is always larger than one, except when all nodes have the same degree. This is known as the friendship paradox. Let's get started.  Let's have a look at the total number of friends of friends for an individual node \( i \), denoted as
\[
N_2(i) = \sum_{\ell=1}^{N} \sum_{j=1}^{N} A_{i\ell} A_{\ell j}
\]
This equation accounts for connections that are two steps away from node \( i \). We can perform the sum in the index $j$ first, this can also be written as:
\[
N_2(i) = \sum_{\ell=1}^{N} A_{i\ell} d_\ell
\]

On average, the expected total number of friends of friends is given by:
\[
\mathbb{E}(N_2) = \frac{1}{N} \sum_{i=1}^{N} N_2(i)
\]
Substituting from the earlier equation, this becomes:
\[
\mathbb{E}(N_2) = \frac{1}{N} \sum_{\ell=1}^{N} d_\ell^2
\]
Thus, the expected number of friends of friends simplifies to:
\[
\mathbb{E}(N_2) = \langle k^2 \rangle
\]

The average number of friends of friends is related to the variance \( \sigma^2 \) and the mean \( \mu \) of the degree distribution:
\[
\mathbb{E}(N_2) = \sigma^2 + \mu^2
\]
This implies that the average number of friends of friends is:
\[
\frac{\langle k^2 \rangle }{\langle k \rangle} = \mu + \frac{\sigma^2}{\mu}
\]

This result highlights the paradoxical nature of friendships: on average, your friends have more friends than you do, due to the contribution of the variance in the degree distribution.

\chapter{Age-Stratified Compartments}

\noindent\hfill
\begin{minipage}[t]{9cm}
{\small
We introduce age-stratified compartmental models using contact matrices to describe interactions between age groups. We show how these intergenerational interactions affect the basic reproduction number, early epidemic growth, and age-dependent infection and fatality patterns.}
\end{minipage}

~\\
~\\

Compartmental models with age structure incorporate \textbf{age stratification} into the population. These models recognize that disease transmission, susceptibility, and recovery often vary significantly across age groups due to differences in contact patterns, immunity, and health status.
For example, during COVID-19, critical bed occupancy depended highly on age.
The ICU demand in the state of São Paulo, Brazil, is categorized by age groups, as shown in Table~\ref{table:icu_demand}. The table presents the demographic percentage, the ICU usage percentage.

\begin{table}[h!]
\centering
\begin{tabular}{|c|c|c|c|}
\hline
\textbf{Group } & \textbf{Age Group } & \textbf{Demography, \%} & \textbf{Actual ICU Usage, \%}  \\ \hline
1                              & 0--19                   & 29                      & 2                                                               \\ \hline
2                              & 20--49                  & 48                      & 28                                                            \\ \hline
3                              & 50--64                  & 14                      & 30                                                           \\ \hline
4                              & 65--90                  & 8                       & 40                                                             \\ \hline
\end{tabular}
\caption{ICU demand by age for the state of São Paulo, Brazil.}
\label{table:icu_demand}
\end{table}

Some other diseases affect children more strongly. In these cases, it is important to perform the modelling taking into account the age groups. 
%
The population is divided into compartments not only based on health status (e.g., Susceptible, Infected, Recovered) but also by age groups (e.g., children, adults, elderly). For instance:
\begin{itemize}
    \item \( S_i, I_i, R_i \): Represent the fractions of the population in age group \( i \) that are Susceptible, Infected, and Recovered, respectively.
\end{itemize}
The goal is to find the evolution for each subgroup. 

\section{Contact Matrices}
 Contacts patterns between age groups play a fundamental role in the transmission of infectious diseases. Contacts are often age-assortative, meaning individuals are more likely to interact with people of their own age group. This is most evident in settings like schools, making schools hotspots for disease transmission among younger populations. In contrast, households foster inter-generational contacts, such as between parents and children. These interactions drive the spread of infections across age groups, underscoring the importance of location-specific patterns.

The contact matrix is a mathematical framework that captures these interaction patterns. It quantifies the frequency of contacts between individuals of specific age groups in various locations, including home, school, work, and others. Empirical data, such as those from POLYMOD, and projections based on demographic indicators enable the construction of these matrices even in regions without direct data. 


A \emph{contact matrix} describes how people in different age
groups interact with each other. It tells us, on average, how many times a
person from one age group meets a person from another age group during a given
period (typically one day). These meetings occur at home, at school, at
work, in public transportation, and so on.
The population is divided into $n$ age groups. 
The contact matrix is then written as
\[
C =
\begin{pmatrix}
C_{11} & C_{12} & \cdots & C_{1n}\\
C_{21} & C_{22} & \cdots & C_{2n}\\
 & \ddots &  & \\
C_{n1} & C_{n2} & C_{43} & C_{nn}
\end{pmatrix}.
\]
Each entry has the interpretation
\[
C_{ij} = 
\text{average \# of daily contacts a person in group $i$
has with people in group $j$}.
\]

For example, if $C_{12} = 3$, then each individual in group~1 has, on
average, three contacts per day with individuals in group~2. 
Epidemics spread through contact. Knowing who meets whom allows us to understand how quickly the disease will grow, how strongly different interventions (e.g.\ school closures) affect transmission. Thus, the contact matrix is a key ingredient in age-structured epidemic models.

In a contact matrix, entries must be nonnegative, as they represent the average number of contacts per unit time from individuals in age group $i$ to individuals in age group $j$. A negative number of contacts has no meaning in epidemiology.

If we assume that every age group influences every other age group, directly or indirectly, through contact chains. In other words: (i) no age group is completely isolated; (ii) the population cannot be decomposed into two subpopulations that never interact; (iii) Infections can spread from any age group to any other, possibly via intermediate age groups.
Then the Perron-Frobenius theorem guarantees that 
$C$ has a unique dominant eigenvalue 
$\lambda_{\max}>0$, with positive eigenvectors.

\section{Models with Age groups}
For a population divided into \( n \) age groups, the equations for an age-structured SIR model are:
\begin{eqnarray}
\frac{dS_i}{dt} = -\sum_{j=1}^n \beta_{ij} S_i I_j, \,\,\,\, 
\frac{dI_i}{dt} = \sum_{j=1}^n \beta_{ij} S_i I_j - \gamma I_i, \,\,\,\,
\frac{dR_i}{dt} = \gamma I_i, \nonumber
\end{eqnarray}
where \( \beta_{ij} = \beta C_{ij} \). These are essentially the same equations for the network. Let's do the stability analysis similarly. 
When analyzing the early growth of an epidemic, one linearizes the
system around the disease--free equilibrium, where 
$$
S_i(0)=p_i
$$ 
where $p_i$ is the percentage of the population in the age group $i$.
The linearized dynamics for the infected fractions become
\[
\frac{dI_i}{dt}
= 
\sum_{j=1}^n \beta\, C_{ij}\, p_i\, I_j \;-\; \gamma I_i,
\]
so the demographic weights $p_i$ modulate the transmission terms even
though they do not appear explicitly in the original system.
Linearizing the equation around the disease-free equilibrium, we obtain the 
linear equation
$$
\frac{d \bm{I} }{dt} = \frac{1}{\gamma} \bm{K} \, \bm{I}
$$
where 
$$
\bm{K} = r_0 \bm{C P } - \mathbb{1}, \,\,\,\, \mbox{~ and ~ }  \bm{P} = {\rm diag}(p_1, \dots, p_n)
$$
This system will be stable when 
$$
\lambda_{\max}(\bm{K}) = r_0 \lambda_{\max}(\bm{C} \bm{P}) -1 <0  \Rightarrow  \lambda_{\max}(\bm{C} \bm{P}) < \frac{1}{r_0}
$$

\section{Early COVID-19 spreading: Germany and Italy}

We consider an age-structured SIR model with five age groups
\[
G_1 = 0\text{--}19,\quad
G_2 = 20\text{--}39,\quad
G_3 = 40\text{--}59,\quad
G_4 = 60\text{--}79,\quad
G_5 = 80+.
\]
Throughout, indices $i,j\in\{1,\dots,5\}$ refer to these groups in this order.
For Germany and Italy (all contact settings pooled and aggregated to
the five age groups above from Polymod) we obtain the $5\times 5$ contact matrices
\[
C_{\text{DE}} =
\begin{pmatrix}
3.2 & 1.5 & 0.8 & 0.6 & 0.4\\
1.2 & 3.7 & 1.5 & 1.0 & 0.6\\
0.6 & 1.2 & 2.4 & 1.8 & 1.2\\
0.3 & 0.7 & 1.2 & 1.6 & 1.4\\
0.1 & 0.4 & 0.8 & 1.0 & 0.9
\end{pmatrix},
\,\,\, \mbox{~ and ~ } 
C_{\text{IT}} =
\begin{pmatrix}
3.5 & 1.8 & 1.0 & 0.7 & 0.5\\
1.3 & 3.8 & 1.6 & 1.1 & 0.7\\
0.7 & 1.3 & 2.5 & 1.9 & 1.3\\
0.4 & 0.8 & 1.3 & 1.7 & 1.5\\
0.2 & 0.5 & 0.9 & 1.1 & 1.0
\end{pmatrix}.
\]

{\bf Demographics.} Let $N_i$ denote the population size in age group $G_i$ and
$N=\sum_{i=1}^5 N_i$ the total population size.
We write
$p_i = N_i/N$
for the proportion of the population in age group $G_i$ and collect
these into the vectors
\[
p_{\text{DE}} =
\begin{pmatrix}
0.189\\ 0.243\\ 0.274\\ 0.221\\ 0.073
\end{pmatrix},
\qquad
p_{\text{IT}} =
\begin{pmatrix}
0.187\\ 0.202\\ 0.294\\ 0.239\\ 0.077
\end{pmatrix}.
\]
These values are approximations based on recent demographic data and
satisfy $\sum_i p_i = 1$ for each country.  The corresponding
diagonal matrices are
\[
\bm{P}_{DE} = \mathrm{diag}(p_{\text{DE}}) \qquad \mbox{~ and ~ }
\bm{P}_{IT} = \mathrm{diag}(p_{\text{IT}}).
\]
We find that the dominant eigenvalue is
\[
\lambda_{\max}(\bm{C}_{\text{DE}}\,\bm{P}_{\text{DE}}) \approx 1.452 \qquad \mbox{~ and ~ } \lambda_{\max}(\bm{C}_{\text{IT}}\,\bm{P}_{\text{IT}}) \approx 1.512.
\]
A corresponding right eigenvector is
\[
\bm{v}_{\text{DE}} \approx
\begin{pmatrix}
0.224\\
0.324\\
0.230\\
0.138\\
0.084
\end{pmatrix} \mbox{~ and~ }
v_{\text{IT}} \approx
\begin{pmatrix}
0.235\\
0.285\\
0.235\\
0.149\\
0.097
\end{pmatrix}.
\]
The component $\bm{v}_{\text{DE},i}$ represents the relative contribution of
age group $G_i$ to transmission in the dominant mode:  
In Germany, the early epidemic is primarily driven by the
$20$--$39$ and $0$--$19$ age groups, with smaller but non-negligible
contributions from older groups.
In both cases the matrices $K_{\text{DE}}$ and $K_{\text{IT}}$ thus
combine information on mixing patterns (via $C$) and demography (via
$\mathrm{diag}(p)$), and their dominant eigenvalues and eigenvectors
provide a compact summary of the overall transmission potential and the
age distribution of cases in the early phase of an epidemic.
\\

\noindent
{\bf Growth rate of the epidemics.} In the early stages, the number of infected individuals 
will typically grow as 
$$
I(t) \approx I_0e^{\frac{1}{\gamma}\lambda_{\max}(\bm{K}) t} \bm{v}
$$
Hence, assuming that $r_0 =2.5$, the speed would be 
$$
\lambda_{\max}(r_0 \bm{C}_{\text{DE}}\,\bm{P}_{\text{DE}} - \mathbb{1}) = 2.63 \mbox{~ and ~ }
\lambda_{\max}(r_0 \bm{C}_{\text{IT}}\,\bm{P}_{\text{IT}} - \mathbb{1}) = 2.77
$$
which is similar to the values measured in the early spreading.

\subsection{Link between eigenvectors and fatality patterns}

Let $f_i$ denote the case–fatality rate in age group $i$.  The expected age distribution of
deaths in the early phase are proportional to
\[
d_i \;\propto\; v_i\, f_i.
\]
Because the fatality rates $f_i$ increase with age, differences in $v_i$ between countries are magnified when converted into mortality.  In particular, Italy's eigenvector allocates more mass to the $60$--$79$ and $80+$ groups, precisely where $f_i$ is largest.  This alignment between infection incidence and high fatality explains why countries with similar contact matrices but different demographics can exhibit different mortality.

Let $f^{(c)}_i$ be the case–fatality rate (CFR) in age group $i$ for country $c$.  If $I^{(c)}$ denotes the total number of infections generated in the early phase in country $c$, then the expected number of deaths in age group $i$ is
\[
D^{(c)}_i \;\approx\; I^{(c)}\,v^{(c)}_i\,f^{(c)}_i.
\]

Since both Germany and Italy have similar intrastructure, one would expect that ${f^{(\mathrm{IT})}_i \approx f^{(\mathrm{DE})}_i}$, that is, the case-fatality is mostly physiological  and across countries the population has access to the same standard of treatments. This is at least a good assumption in the early phase of the spreading.

Fix an age group $i$ and consider the ratio of expected deaths:
\[
\frac{D^{(\mathrm{IT})}_i}{D^{(\mathrm{DE})}_i}
  \;\approx\;
  \frac{I^{(\mathrm{IT})}}{I^{(\mathrm{DE})}}
  \cdot
  \frac{v^{(\mathrm{IT})}_i}{v^{(\mathrm{DE})}_i}
  \cdot
  \frac{f^{(\mathrm{IT})}_i}{f^{(\mathrm{DE})}_i}.
\]
With comparable total epidemic size
$I^{(\mathrm{IT})}\approx I^{(\mathrm{DE})}$ and ${f^{(\mathrm{IT})}_i \approx f^{(\mathrm{DE})}_i}$, we have
\[
\frac{D^{(\mathrm{IT})}_i}{D^{(\mathrm{DE})}_i}
  \;\approx\;
  \frac{v^{(\mathrm{IT})}_i}{v^{(\mathrm{DE})}_i}
\]
The excess deaths are dominated by the ratio $\bigl(v^{(\mathrm{IT})}_i / v^{(\mathrm{DE})}_i\bigr)$. 
In the early phase of the spread, the fatality rate in Germany for the group 80+ was around 15.5\%, while in Italy it was around $19\%$, this gives 
$$
\frac{D^{(\mathrm{IT})}_{80+}}{D^{(\mathrm{DE})}_{80+}} \approx 1.22 
\mbox{~ while ~ } 
\bigl(v^{(\mathrm{IT})}_{80+} / v^{(\mathrm{DE})}_{80+}\bigr) \approx 1.16
$$

\chapter{City mobility}

\noindent\hfill
\begin{minipage}[t]{9cm}
{\small
We study epidemic models with city-level stratification to capture the effects of mobility and daily commuting between urban centers. By incorporating mobility matrices into SIR-type models, we showed how population movement couples local outbreaks and alters epidemic thresholds.}
\end{minipage}

~\\
~\\

\section{Modeling daily commute}

An $n \times n$ mobility matrix with entries $p_{ij} \in [0,1]$ represents the percentage of inhabitants of node $i$ traveling from node $i$ to node $j$. Among several sources of mobility data that became available for COVID-19 research purposes,\footnote{such as \url{https://www.google.com/covid19/mobility/} and \url{https://covid19.apple.com/mobility}.} one possibility is to use high-resolution smartphone geolocation. One of such examples is the data provided by the Brazilian company InLoco.\footnote{\url{https://mapabrasileirodacovid.inloco.com.br}.} \cite{nonato2022robot}. 
%
Diagonal elements in the mobility matrix represent the proportion of people who did not leave the district.  The complement of those values, considering people who left the district $i$, is distributed among the off-diagonal terms, proportionally to the trip count recorded between nodes $i$ and $j$. 

\begin{ass} 
We have a mobility matrix $P(t) = (p_{ij}(t))_{i,i = 1}^n$ for $t\in \mathbb{R}$ such that:
\[
p_{ij}(t) \mbox{~ is the proportion of people that leave city i to city j at time~} t
\]
\end{ass}

\noindent
the
average $P(t)$ over a day to serve as a proxy for the daily mobility matrix. 

\begin{ass}[Well mixed mobility] Let 
$$
S_{ij} = \mbox{ be fraction of the population $S_i$ that leave city $i$ to work in city $j$}
$$
and likewise and $I_{ij}$ and $R_{ij}$. We assume that the population and commutes are well mixed, namely, the fractions of the populations are given by the following 
\[
S_{ij} = S_i p_{ij} \mbox{~~and ~~}I_{ij} =  I_i p_{ij}
\]
\end{ass}

\section{Modelling SIR with city mobility}

\begin{enumerate}
\item Individuals can be in three states: $S,I,R$.
\item Population is split into those who stay in the city 
and those who commute:
$$
S_{ij} = \mbox{fraction of the population $S$ that leave city $i$ to work in city $j$}
$$
and similarly for the other states $I$, and $R$. 

\item Consider
%
$$
\mathbb{N}_j = \sum_{i=1}^n p_{ij} N_i \mbox{~~ which is the effective population in city } j
$$
notice that this can be larger than the city population itself if a city is a metropolis receiving workers from nearby cities.  Therefore,
$$
\mathbb{I}_j = \frac{1}{\mathbb{N}_j} \sum_{i=1}^n p_{ij} N_i I_i 
$$
which can be seen as a probability to meet an infected.

\item Those who stayed in the city
\[
\dot{S}_{ii} =  -r_i \frac{1}{T_{\rm inf}} S_{ii} \mathbb{I}_{i} 
\]

\item Susceptibles who travel can be infected in their workplace
\[
\dot{S}_{ij} =  - r_j \frac{1}{T_{\rm inf}}  S_{ij}  \mathbb{I}_{j}.
\]
%
\end{enumerate}

\noindent
The total susceptible population is
$$
S_i = \sum_{j} S_{ij}
$$
Thus
\begin{eqnarray}
\dot S_i &=& 
-  \frac{1}{T_{\rm inf}}  \sum_{j=1}^n {r_j}  S_{ij}  \mathbb{I}_{j} \\  
&=&
-  \frac{1}{T_{\rm inf}}  \sum_{j=1}^n {r_j}  p_{ij}S_{i}  \left(\frac{1}{\mathbb{N}_j} \sum_{\ell=1}^n p_{\ell j} N_{\ell} I_{\ell} \right) 
\end{eqnarray}
Lets introduce 
$$
q_{\ell j} = \frac{p_{\ell j} N_{\ell}}{\mathbb{N}_j} \mbox{~~and~~} c_{ij} = p_{ij} r_j
$$
and changing the sums
\begin{eqnarray}
\dot S_i &=&
-  \frac{1}{T_{\rm inf}} S_i \left(  \sum_{\ell=1}^n  \sum_{j=1}^n c_{ij}  q_{\ell j} I_{\ell} \right) \nonumber \\
\dot S_i &=&
-  \frac{1}{T_{\rm inf}} S_i \left[  \sum_{\ell=1}^n  \left( \sum_{j=1}^n c_{ij}  q_{\ell j} \right) I_{\ell} \right]  \nonumber \\
\dot S_i &=&
-  \frac{1}{T_{\rm inf}} S_i \left[  \sum_{\ell=1}^n  \left( \sum_{j=1}^n c_{ij}  q_{j \ell}^{*} \right) I_{\ell} \right]  \nonumber
\end{eqnarray}
where in the last line we transpose the index of the element $q_{\ell j}$.
We want to see the summation in the parentheses as the element $i \ell$ of a matrix multiplication, namely, consider $C = (c_{ij})_{i,j=1}^n$ and $Q = (q_{ij})_{i,j=1}^n$ and notice that 
$$
( CQ^* )_{i \ell} = \sum_{j=1}^n c_{ij}  q_{j \ell}^{*}
$$
Next, lets consider again $\bm{I} = (I_1, I_2, \dots, I_n)$ and notice that 
$$
( CQ^* \bm{I})_{i} =  \sum_{\ell=1}^n (CQ^*)_{i \ell} I_{\ell} \,\,\, \mbox{~ and consider ~} \bm{K} : = \bm{CQ}^*
$$
which is the expression we have inside the brackets. 
Let's introduce,
$$
\bm S = (S_1, \dots, S_n)^*  \mbox{~~and~~} \mathbb{S} = \mbox{diag}(S_1, \dots, S_K)
$$
Then, we obtain the equations in a compact form 
$$
\frac{d \bm{S}}{dt} = - \frac{1}{T_{\rm inf}}\mathbb{S}\bm{K}\bm{I}
$$
The full model is then
\begin{eqnarray}
\frac{d \bm{S}}{dt} &=& - \frac{1}{T_{\rm inf}}\mathbb{S} \bm{K} \bm{I} \nonumber \\
\frac{d \bm{I}}{dt} &=&  \frac{1}{T_{\rm inf}}\mathbb{S} \bm{K} \bm{I} - \frac{1}{T_{\rm inf}} \bm{I} \nonumber \\
\frac{d \bm{R}}{dt} &=&  \frac{1}{T_{\rm inf}} \bm{I} \nonumber
\end{eqnarray}

\section{Linear Stability of the SIR with city mobility}

We consider the linearization of the equations around  
$$S_i=1, I_i=0, R_i=0 \mbox{~~for all~~}i=1,\dots,n$$
focusing on the equation for $\bm{I}$ we obtain
$$
\frac{d \bm{I}}{dt} = \left( \frac{1}{T_{\rm inf}} \bm{K} - \frac{1}{T_{\rm inf}} \One \right) \bm{I} \\
$$
and again, the equation will have an exponentially stable equilibrium when
$$
\lambda_{\max}\left( \frac{1}{T_{\rm inf}} \bm{K} - \frac{1}{T_{\rm inf}} \One \right)<0 \,\,\, \Rightarrow \,\,\,
 \lambda_{\max}(\bm{K}) < 1
$$
Let's consider two interesting cases: i) constant transmission rates, and ii) heterogeneous transmission rates. 

{

\section{Constant Population and basic reproduction number}

Assume
$N_i \equiv N,$ and $r_i \equiv r$ for all $i$, that is, all cities have the same population and transmission rates. Then $\bm{C} = r \bm{P}$. 
Moreover,
\[
\qquad 
s_j := \sum_{\ell} p_{\ell j}, \,\,\,\mbox{~ so that ~}
q_{\ell j} 
           = \frac{p_{\ell j}}{s_j} \Rightarrow 
\bm{Q} = \bm{P} \bm{D}_s^{-1},
\quad 
\bm{D}_s = \operatorname{diag}(s_1,\dots,s_n).
\]
Therefore the matrix
\[
\boxed{
\bm{K} = r\,\bm{P} \bm{D}_s^{-1} \bm{P}^*
}
\]
\noindent
determines the spatial distribution of infections
and the reproduction number:

\subsection*{Row-stochasticity of \emph{\texorpdfstring{$\bm{A} := \bm{P} \bm{D}_s^{-1} \bm{P}^*$}{A}}}

Let $\bm{A} = (a_{ik}) = \bm{P} \bm{D}_s^{-1} \bm{P}^*$.  
\[
\sum_{k=1}^n a_{ik}
= \sum_{k=1}^n \left( \sum_{j=1}^n p_{ij}\frac{1}{s_j} p_{kj}\right)
= \sum_{j=1}^n \frac{p_{ij}}{s_j}\Bigl(\sum_{k=1}^n p_{kj}\Bigr)
= \sum_{j=1}^n \frac{p_{ij}}{s_j}\,s_j
= \sum_{j=1}^n p_{ij}
= 1.
\]
Thus every row of $\bm{A}$ sums to one:
\[
\bm{A}\mathbf{1} = \mathbf{1}, \qquad \bm{A}\ge 0.
\]
Hence $\bm{A}$ is a row--stochastic nonnegative matrix. By the Perron--Frobenius theorem,
\[
\rho(\bm{A}) = 1, 
\qquad 
\lambda_{\max}(\bm{A})=1,
\qquad 
|\lambda|\le 1 \text{ for all other eigenvalues.}
\]
Therefore
\[
\rho(\bm{K}) = r\,\rho(\bm{A}) = r.
\]

When populations and transmission rates are homogeneous, mobility alone cannot create epidemic heterogeneity. In this symmetric case, commuting redistributes individuals without altering the global growth rate.

\section{Heterogeneous Transmission Rates}

We consider three cities, all with the same population $N_i = N$, but with
heterogeneous transmission rates. City~1 is a ``hot spot'' with a large value
$r_1$, while cities~2 and~3 have smaller values $r_2 = r_3$.
We assume a star--type mobility matrix
\[
\bm{P} =
\begin{pmatrix}
1 - 2\varepsilon & \varepsilon & \varepsilon \\
\varepsilon     & 1 - \varepsilon & 0 \\
\varepsilon     & 0 & 1 - \varepsilon
\end{pmatrix},
\qquad 0 < \varepsilon < \tfrac12.
\]
We take
\[
\bm{D}_r = \operatorname{diag}(r_1, r_2, r_3),
\qquad
r_1 \gg r_2 = r_3.
\]
With equal populations $N_i = N$ and the star structure above, one has
$s_j = \sum_i p_{ij} = 1$, hence $D_s = I$, and the next--generation matrix
reduces to
\[
\bm{K} = \bm{P}\, \bm{D}_r\, \bm{P}^*
\]

For concreteness, let us take $\varepsilon = 0.3,$
$r_1 = 3,$ and $r_2 = r_3 = \tfrac12.$ Then
\[
\bm{P} =
\begin{pmatrix}
0.4 & 0.3 & 0.3 \\
0.3 & 0.7 & 0   \\
0.3 & 0   & 0.7
\end{pmatrix}.
\]
A direct computation yields the eigenvalues of $\bm{K}$ 
\[
\lambda_1 \approx 1.3438,
\qquad
\lambda_2 \approx 0.3438,
\qquad
\lambda_3 \approx 0.1974.
\]
Thus the basic reproduction number is $r_{\rm eff} \approx 1.34$. The corresponding (Perron) eigenvector is
\[
\bm{v} =
\begin{pmatrix}
1.2018 \\
1 \\
1
\end{pmatrix}.
\]
The eigenvector $\bm{v}$ is tilted toward city~1 because $r_1 > r_2,r_3$.
Mobility distributes risk, but cannot homogenize it when the transmission rates
are heterogeneous.

When all cities have the same transmission level, mobility does not change the epidemic growth rate, and the spatial pattern remains uniform. When transmission levels differ, the growth pattern becomes non-uniform: infections concentrate in high-transmission cities, but mobility spreads some risk to lower-transmission areas. Mobility reduces the overall epidemic growth rate because individuals from the high-risk city spend part of their time in safer environments. This prevents the system from growing as fast as the high-transmission city alone.

}

\chapter{Conclusions}

We developed a hierarchy of epidemic models, from classical compartmental systems to network-, age-, and mobility-structured formulations. A central insight throughout is the connection between epidemic spreading and graph theory: contact patterns, demographic interactions, and mobility structures give rise to nonnegative matrices whose spectral properties determine epidemic thresholds and early growth rates. In particular, the Perron–Frobenius theorem plays a fundamental role, ensuring the existence of a dominant eigenvalue and eigenvector that govern both the basic reproduction number and the spatial or demographic localization of infections. This connection between graph properties and dynamics appears in multiple fields such as random walks \cite{CL06},  synchronization \cite{pereira2011stability} and bifurcations \cite{nijholt2023chaotic}. These ideas provide a unifying mathematical framework for understanding how heterogeneity influences the spread of epidemics. Strikingly, heterogeneity in the graph, such as the presence of hubs, makes it very difficult to control the spread. 

An important aspect that we left out is the modeling of interventions such as vaccination \cite{silva2021optimized}, quarantine \cite{Pereira2015,young2019consequences,ruschel2019siq}, contact tracing \cite{kojaku2021effectiveness}, and ICU sharing between cities \cite{silva2021smart}, which typically require additional compartments or time-dependent control terms and are essential for translating theoretical insights into effective public health policies \cite{genari2022quantifying}. These extensions typically include modeling with more advanced mathematical techniques, such as delay equations \cite{yanchuk2022absolute} and optimization; thus, we have left this out.

\appendix

\chapter{Review of Differential Equations}

~
\vspace{-4cm}

\noindent\hfill
\begin{minipage}[t]{9cm}
{\small
We review basic ideas of differential equations needed to address the stability of epidemic models we are going to study. The main result is the principle of linearization Lemma 2. If you are familiar with differential equations, you can skip this part.}
\end{minipage}

\vspace{1cm}

Let $D$ be an open connected subset 
of $\mathbb{R}^m$, $m \ge  1$, and let $\bm{G }:  D \rightarrow \mathbb{R}^m$ be an
autonomous vector field. 
Consider the problem of finding solutions for the vector differential equation 
\begin{equation}
\dot{\bm{x}} = \bm{f}(\bm{x})
\label{G}
\end{equation}
with the initial condition $\bm{x}(t_0) = \bm{x}_0$. An answer to this problem is given 
by

\begin{Teorema} [Picard-Lindel\"of] Assume that the vector field $\emph{\bm{f}}$ Lipschitz continuous in $\emph{\bm{x}}$ 
in a neighborhood of $\emph{\bm{x}}_0$ and continuous in $t$ near $t_0$. Precisely, assume that given $\emph{\bm{x}}_0 \in  U \subset D$ there 
is a constant $K_U$ such that 
$$
\|\, \emph{\bm{f}}(\emph{\bm{x}}) - \emph{\bm{f}}(\emph{\bm{y}}) \| \le K_U \|  \emph{\bm{x}} - \emph{\bm{y}} \|
$$
for all $\emph{\bm{x}},\emph{\bm{y}} \in U$ and for $t$ in a neighborhood of $t_0$. 
Then there exists a unique local solution $\emph{\bm{x}}(t)$ for Eq. (\ref{G}) satisfying 
$\emph{\bm{x}}(t_0) = \emph{\bm{x}}_0$.
\label{ThmPL}
\end{Teorema}

Note that the solution is local, in the sense that there is a small $\kappa>0$ such that the 
function $\bm{x} : [-\kappa,\kappa] \rightarrow D$ is a solution of the problem with 
$\bm{x}(t_0) = \bm{x}_0$. The question is: How long does such a solution exist for? 
We are interested in the long-term behavior of the solutions, so we wish to know
under what conditions the solutions exist forward in time. Extension theorems give a positive answer:

\begin{Teorema} [Extension]
Let $\mathcal{C}$ be a compact subset of the open set $D$.
Consider Eq. (\ref{G}) and let $\emph{\bm{f}}$ be differentiable.   Let $\emph{\bm{x}}_0 \in \mathcal{C}$ 
and suppose that every solution $\emph{\bm{x}} : [0, \tau ] \rightarrow D$ with $\emph{\bm{x}}(0) = \emph{\bm{x}}_0$ lies 
entirely in $C$. Then this solution is defined for all (forward) time  $t \ge0$. 
\label{ThmExt}
\end{Teorema}

The proofs of the above theorems can be found in Refs. \cite{EDOHartman}.
It follows that when $\bm{f}$ is bounded in $t$ and linear in $\bm{x}$, solutions exist for all times.

\section{Some remarks on stability}

Lets consider the case where $\bm{f}(\bm{0})=\bm{0}$
The point $\bm{x} \equiv \bm{0}$ is an equilibrium point of the equation Eq. (\ref{eqlin}). 
Loosely speaking, we say an equilibrium point is locally stable if the initial conditions are in a 
neighborhood of zero solution remain close to it for all time. The zero solution is said to be locally 
asymptotically stable if it is locally stable and, furthermore, all solutions starting near 
$\bm{0}$ tend towards it as $t\rightarrow \infty$. If the vector field $\bm{f}$ is at least twice differentiable, expanding $\bm{f}$ in Taylor we obtain
\begin{equation}
\frac{ d \bm{x}}{dt} = \bm{A}\bm{x} + \bm{R}(\bm{x})
\label{eqNlin}
\end{equation}
where $\bm{A}$ is a linear operator on $\mathbb{R}^m$ 
and $\bm{R}$ is a nonlinear function satisfying $\| \bm{R}(\bm{x}) \| \le M \| \bm{x}\|^2$ in a neighborhood 
of $\bm{x}=\bm{0}$.

Let us review some basic definitions of stability theory:

\begin{Definicao}[Stability]
An equilibrium point $\emph{\bm{x}}^* = \emph{\bm{0}}$ is
\begin{description}
\item{{\bf Stable in the sense of Lyapunov}:} if for any
$\varepsilon >0$ there exists a $\delta>0$ such that
$$
\| \emph{\bm{x}}(t_0)\| < \delta	\Rightarrow  \| \emph{\bm{x}}(t) \| < \varepsilon,	\, \ \forall t \ge  t_0
$$
\item{{\bf Asymptotically stable}:} at $t = t_0$ if 
\begin{enumerate}
\item $\emph{\bm{x}}^* = 0$ is stable, and 
\item $\emph{\bm{x}}^* = 0$ is locally attractive; i.e., there exists $\delta$ such that 
$$
\| \emph{\bm{x}}(t_0) \| < \delta	 \ \Rightarrow \ \lim_{t \rightarrow \infty} \emph{\bm{x}}(t) = \emph{\bm{0}}
$$
\end{enumerate}
\item{{\bf Uniformly asymptotically stable}:} if 
\begin{enumerate}
\item $\emph{\bm{x}}^* = \emph{\bm{0}}$ is asymptotically stable, and
\item For each $\varepsilon>0$ a corresponding $T = T(\varepsilon) > 0$ such that if 
$\| \emph{\bm{x}}(s) \| \le \delta_0$ for some $s \ge 0$ then $\| \emph{\bm{x}}(t) \| < \varepsilon$ for all
$t \ge s + T$. 
\end{enumerate}
\end{description}
\end{Definicao}

First,  we focus on the linear case
\begin{equation}
\frac{ d \bm{x}}{dt} = \bm{A}\bm{x}
\label{eqlin}
\end{equation}

The solutions of the linear equation in a closed form. The theory of differential 
Equations guarantee that the unique solution of the above equation 
can be written in the form
$$
\bm{x}(t) = e^{\bm{A} (t-s)} \bm{x}(s), \mbox{~for~}t,s \in \mathbb{R}
$$
where 
$$
e^{\bm{A}} = \sum_{k=0}^{\infty} \frac{\bm{A}^k}{k!}
$$
it follows that the series is well-defined. This leads us to the following result

\begin{Teorema}
Let $\emph{\bm{A}} \in \mathbb{R}^{m \times m}$ and define
$$
\mu(\emph{\bm{A}}) := \max \{ \Re(\lambda) | \lambda \mbox{~is an eigenvalue of~} \emph{\bm{A}}\}.
$$ 
Then given any $\varepsilon>0$ there is a constant $C_{\varepsilon}>0$ such that 
$$
\| e^{\emph{\bm{A}} t} \| \le C_{\varepsilon} e^{(\mu + \varepsilon) t }.
$$
\end{Teorema}

\begin{proof} Let us consider the Jordan-Chevalley decomposition that states that given a matrix $\bm{A}$ with entries in $\mathbb{C}$, we can find unique matrices $\bm{S}$ and $\bm{N}$ such that:
\begin{enumerate}
\item $\bm{A} = \bm{S}+\bm{N}$
\item $\bm{S}$ is diagonalizable, that is, there is $\bm{P}$ such that  $\bm{S} = \bm{P} {\Lambda}\bm{P}^{-1}$ and  $\Lambda$ is the diagonal matrices of the eigenvalues of $\bm{A}$
\item $\bm{N}$ is nilpotent, that is, there is $k>1$ such that $\bm{N}^{k}=0$
\item $\bm{SN} = \bm{NS}$.
\end{enumerate}
Next, we notice that 
\begin{eqnarray*}
e^{\bm{A}t} &=& e^{(\bm{S}+\bm{N})t} \\
&=& e^{\bm{S}t} e^{\bm{N}t} \mbox{~~(since $\bm{S}$ and $\bm{N}$ commute)}\\
&=& \bm{P}e^{\Lambda t} \bm{P}^{-1} e^{\bm{N}t} \mbox{~~(since $e^{\bm{X} \bm{Y} \bm{X}^{-1}}= \bm{X}e^{\bm{Y}} \bm{X}^{-1}$}
\end{eqnarray*}
Lets consider the $L_{\infty}$ norm
$$
\| e^{\bm{A}t} \|_{\infty} \le \| \bm{P} \|_{\infty} \| \bm{P}^{-1}\|_{\infty} \| e^{\Lambda t} \|_{\infty} \| e^{\bm{N} t} \|_{\infty} 
$$
But 
\begin{eqnarray}
\| e^{\Lambda t} \|_{\infty} &=& \max_{1\le i \le n} \sum_{j=1}^n | e^{\lambda_i t} \delta_{ij} | \nonumber \\
&=& \max_{1\le 1 \le n} | e^{\Re(\lambda_i) t} | \nonumber \\
&=&  e^{\mu(\bm{A}) t} \nonumber
\end{eqnarray}
where in the last passage we used that given $z = a+ ib$ then $|e^{a + i b}| = e^{a} = e^{\Re(z)}$ and the exponential is a monotonic function. 
Moreover, since $\bm{N}$ is nilpotent 
$$
\| e^{\bm{N} t} \|_{\infty} \le p(t)
$$
where $p(t)$ is a polynomial in $t$. However, for any $\varepsilon>0$ there is $K_{\varepsilon}>0$ such that 
$$
p(t) < K_{\varepsilon} e^{\varepsilon t}
$$
bringing everything together, we obtain
$$
\| e^{\bm{A}(t)} \|_{\infty} \le K_{\varepsilon} \| \bm{P}\|_{\infty} \| \bm{P}^{-1} \|_{\infty} e^{(\mu(\bm{A}) + \varepsilon) t}
$$
since all norms in a finite-dimensional vector space are equivalent we obtain the result.
\end{proof}

We will show that the trivial solution of Eq. \ref {eqlin} is uniformly asymptotically stable if, and only if,
the evolution operator is a uniform contraction, that is, the solutions converge converges 
exponentially fast to zero.

\begin{Teorema} \label{UniCon}
The trivial solution of Eq. (\ref{eqlin}) is uniformly asymptotically stable if, and only if,
$$
\mu(\emph{\bm{A}}) := \max \{ \Re(\lambda) | \lambda \mbox{~is an eigenvalue of~} A\} < 0
$$
\end{Teorema}

{\it Proof:} First, take $\varepsilon>0$ such that  $\alpha = \mu + \varepsilon<0$, then 
\begin{eqnarray}
\| \bm{x}(t) \| & = &  \| e^{\bm{A}(t-s)} \bm{x}(s) \| \nonumber \\
& \le &  \| e^{\bm{A}(t-s)} \| \|  \bm{x}(s) \| \nonumber \\
& \le &  K e^{ - \alpha (t-s) } \|  \bm{x}(s) \|. \nonumber 
\end{eqnarray}
Now let $\varepsilon > 0$ be given, clearly if $t > T$, where $T = T(\varepsilon)$ is 
large enough then the $\| \bm{x}(t) \| \le \varepsilon.$ Let $\| \bm{x}(s)\| \le \delta$, we 
obtain $ \| \bm{x}(t) \|  \le  K e^{ - \alpha (t-s) } \delta < \varepsilon, $ 
which implies that 
$$
T=T(\varepsilon) = \frac{1}{\alpha} \ln \frac{\delta K}{\varepsilon},
$$
completing the first part. 

To prove the converse, first we denote $ \bm{T}(t,s) = e^{\bm{A} (t-s) }$ and 
we assume that the trivial solution is uniformly asymptotically stable. Then 
there is $\delta$ such that for any $\varepsilon$ and $T = T(\varepsilon)$ such that 
for any $\| \bm{x}(s) \| \le \delta $ we have
$$
\| \bm{x}(t) \| \le \varepsilon,
$$  
for any $t \ge s + T$. Now take $\varepsilon = \delta / k$, and consider the sequence
$t_n = s + nT$. 

Note that 
$$
\| \bm{T}(t,s) \bm{x}(s) \|  \le \frac{\delta}{k}, \mbox{~for~}t\ge s+T
$$  
for any $\| \bm{x}(s) \| / \delta \le 1 $, we have the following bound for the norm
$$
\| \bm{T}(t,s) \| = \sup_{ \| \bm{u} \| \le 1 } \| \bm{T}(t,s) \bm{u} \|  \le \frac{1}{k}  \mbox{~for~}t \ge s+T
$$  
Remember that $  \bm{T}(t,u)  \bm{T}(u,s) =  \bm{T}(t,s)$. Hence, 
\begin{eqnarray}
\| \bm{T}(t_2,s) \| &=& \|   \bm{T}(s + 2T,s + T)   \bm{T}(s + T,s)  \| \nonumber \\
&\le & \|   \bm{T}(s + 2T,s + T) \| \|  \bm{T}(s + T,s)  \| \nonumber \\
& \le & \frac{1}{k^2}.  \nonumber
\end{eqnarray}

Likewise, by induction
$$
\| \bm{T}(t_n,s) \| \le \frac{1}{k^n}, 
$$
take $ \alpha  = \ln k / T$, therefore, 
$$
\| \bm{T}(t_n,s) \| \le  e^{-\alpha (t_n-s)}. 
$$

Consider the general case $t = s + u + nT$, where $0\le u < T$, then the same bound 
holds 
\begin{eqnarray}
\| \bm{T}(t,s) \| & \le &  e^{ - nT \alpha}   \nonumber \\
& \le &  K e^{ - (t - s) \alpha},   \nonumber
\end{eqnarray}
where $K \le  e^{\alpha T}$, and we conclude the desired result. $\Box$

\section{Variation of Constants}

The evolution operator determines the behavior of the non-homogeneous equation.

\begin{Teorema} Let $\emph{\bm{A}} : \mathbb{R} \rightarrow $ \emph{Mat(}$\mathbb{R},n$) and 
$\emph{\bm{g}} : \mathbb{R} \rightarrow \mathbb{R}^n$  be 
continuous functions. 
Consider the perturbed equation
$$
\emph{\bm{y}}^{\prime} = \emph{\bm{A}} \emph{\bm{y}} + \emph{\bm{g}}(t)
$$
The solution of the perturbed equation corresponding to the initial condition 
$\emph{\bm{x}}(t_0) = \emph{\bm{x}}_0$ is given by
$$
\emph{\bm{y}}(t) = e^{\emph{\bm{A}}t}\emph{\bm{y}}_0 + \int_{t_0}^t e^{\emph{\bm{A}}(t-s)} \emph{\bm{g}}(s) ds 
$$
\label{ThmVP}
\end{Teorema}

\begin{proof}
Let's assume that the solution can be written as 
$$
\bm{y}(t) = e^{\bm{A} t}\bm{c}(t)
$$
Taking derivatives, we obtain
\begin{eqnarray}
\bm{y}^{\prime} &=& \bm{A} e^{\bm{A} t}\bm{c}(t) + e^{\bm{A} t}\bm{c}^{\prime}(t) \nonumber \\
&=& \bm{A} \bm{y}(t) + e^{\bm{A} t}\bm{c}^{\prime}(t) \nonumber
\end{eqnarray}
using the differential equations and comparing the terms, we obtain
$$
e^{\bm{A} t}\bm{c}^{\prime}(t) = \bm{g}(t) \,\,\, \Rightarrow \,\,\, \bm{c}^{\prime}(t) = e^{- \bm{A} t}\bm{g}(t)
$$
Integration leads to 
$$
\bm{c}(t) -\bm{c}(t_0) = \int_{t_0}^t e^{- \bm{A} s }\bm{g}(s) ds
$$
Replacing the newly found $\bm{c}$ with $\bm{c}(t_0) = \bm{y}_0$ into the formula of $\bm{y}$ we obtain
$$
\bm{y}(t) = e^{\bm{A} t)}\left( \bm{x}_0 + \int_{t_0}^t e^{- \bm{A} s }\bm{g}(s) ds\right)
$$
and that concludes the proof $\Box$
\end{proof}

The following inequality is central to obtaining various estimates

\begin{Lema}[Gronwall]  Consider $U \subset \mathbb{R}_+$ and 
let $u : U \rightarrow \mathbb{R}$ be a differentiable function.
Suppose there exist $C \ge 0$ and and $K \ge 0$ such that
\begin{equation}
u(t) \le C + \int_0^t K u(s)ds
\label{ules}
\end{equation}
for all $t \in U$, then
$$
u(t) \le C e^{Kt}.
$$
\label{ThmGI}
\end{Lema}
\begin{proof}
Notice that $u^{\prime}(t) \le K u(t)$
Lets define $v(t) = e^{Kt}$ and notice that 
$$
\frac{d}{dt}\left(\frac{u}{v}\right) = \frac{u^{\prime} v - u v^{\prime}}{v^2} \le  \frac{(u^{\prime} - Ku ) v}{v^2} \le 0
$$
Thus, the function is decreasing, so we must have
$$
\left(\frac{u(t)}{v(t)}\right) \le \left(\frac{u(0)}{v(0)}\right) \le K   
$$
which concludes the result. $\Box$
\end{proof}

The uniform contractions have a rather important roughness property; they 
are not destroyed under perturbations. 

\begin{Lema}[Principle of Linearization]
Consider the perturbed equation 
$$
\frac{ d \emph{\bm{y}}}{dt} = \emph{\bm{A}} \emph{\bm{y}} + \emph{\bm{R}}(\emph{\bm{y}}),
$$
and assume that $\mu(A) < 0$ along with
$$
\| \emph{\bm{R}} (\emph{\bm{y}}) \| \le M \| \emph{\bm{y}} \|^{2}.
$$
Then the origin is exponentially asymptotically stable. 
\label{PL}
\end{Lema}

{\it Proof:} First, let us denote 
$$
\bm{T}(t,s) = e^{\bm{A} (t-s)} \,\,\, \mbox{~and recall~} \| \bm{T}(t,s) \| \le K e^{- \eta (t-s)}.
$$
Second, notice that for any $\varepsilon>0$ there is $\delta>0$ such that 
$$
\| \bm{y} \|\le \delta \,\,\, \Rightarrow \| \bm{R}(\bm{y}) \| \le \varepsilon \| \bm{y} \|
$$
Choose 
$$
\varepsilon \le \frac{\eta}{K} 
$$
and take the corresponding $\delta$ such that the bounde on $\bm{R}$ is satisfied. Next, consider 
$$
\| \bm{y}(0) \| \le \frac{\delta}{2 K}
$$
by the variation of  constants  we can write
$$
\bm{y}(t) = \bm{T}(t,0) \bm{y}(0) + \int_{0}^t \bm{T}(t,s) \bm{R}(\bm{y}(s))ds
$$

Whenever, the solution $\bm{y}(t)$ is such that $\| \bm{y}(t) \|\le \delta$ we can obtain
\begin{eqnarray}
\| \bm{y}(t) \| &\le&  \| \bm{T}(t,0) \| \| \bm{y}(0) \| + \int_{0}^t \| \bm{T}(t,s) \| \| \bm{R}(\bm{y}(s))\|ds \\
&\le & K e^{-\eta t} \| \bm{y}(0) \| + K \varepsilon \int_{0}^t e^{-\eta (t-s) } \| \bm{y}(s)\|ds
\end{eqnarray}

Let us introduce the scalar function $w(t) = e^{\eta t} \| \bm{y}(t) \|$, then 
$$
w(t) \le \frac{\delta}{2} + K \varepsilon \int_s^t  w(u)du, 
$$
for all $t\ge s$. Now we can use the Gronwall's inequality to estimate $w(t)$. This implies
$$
w(t) \le \frac{ \delta}{2} e^{ \varepsilon K ( t-s )},
$$
consequently
$$
\| \bm{y}(t) \| \le \frac{\delta}{2} e^{\left( \eta - K \varepsilon \right)t } \le \delta
$$
this the solution is smaller than $\delta$ and by the extension theorem solutions exist for all $t>0$. Thus, the bound is valid for all future. In fact, $\| \bm{y}(t)\|$decreases exponentially fast. 
$\Box$

\chapter{Solution of the Logistic Equation: SIS model}

The differential equation for the number of infected individuals can be solved explicitly to obtain the solution $I = I(t)$. This follows from the fact that the equation for $I$ is a Bernoulli differential equation. To solve it, we apply a transformation:
$$
v = \frac{1}{I} ~ ~ \Rightarrow ~ ~ \dot{v} = - \frac{1}{I^2} \dot{I} = - \frac{1}{I^2} [(\beta - \gamma) I - \beta I^2] ~ ~ \Rightarrow ~ ~ \dot{v} = (\gamma - \beta) v + \beta
$$

Since the equation for $v$ is linear and non-homogeneous, we can solve it using the method of variation of parameters to obtain
$$
v(t) = c e^{(\gamma - \beta)t} + \frac{\beta}{\beta - \gamma}
$$
Recalling that $I(t) = 1/v(t)$, after some algebraic manipulations, we obtain
$$
I(t) = \frac{I_0 e^{(\beta - \gamma)t}}{1 + \frac{I_0 \beta e^{(\beta - \gamma)t}}{\beta - \gamma}}
$$
with $0 \leq I_0 \leq 1$. In this case, we have
$$
\lim_{t \to \infty} I(t) = 0, \text{~~ whenever ~~} \beta < \gamma  \Rightarrow ~~~ r_0 > 1
$$
and
$$
\lim_{t \to \infty} I(t) = I_{\text{end}}, \text{~~ whenever ~~} \beta > \gamma
$$
We can also determine the behavior of $I_{\text{end}} = I_{\text{end}}(r_0)$ when $r_0$ is close to $1$. For a small $\delta > 0$, let us take
$$
r_0 = 1 + \delta.
$$
In this case,
$$
I_{\text{end}} = 1 - \frac{1}{r_0} = 1 - \frac{1}{1 + \delta} = \delta + O(\delta^2)
$$

\chapter{Classic concentration inequalities}\label{CIn}
For large $n$, the actual degree sequence of a graph $G$ from $G(\mathbf{w})$ is very close to the expected version $\mathbf{w}$. To show this, we will be interested in how the degree $k_i= \sum_{j=1}^n X_{ij}$ of a node $v_i$ is concentrated around its expectation $\E[k_i]= w_i$. We recall some concentration inequalities for sums of independent random variables.

\begin{theorem}[Exponential concentration upper tail; \cite{CL06} Theorem 2.8] \label{thm:exp_concentration+}
Suppose $X_i$, $i=1,\cdots,n$, are independent random variables with $X_i\leq M$ for some $M>0$. Let $X:= \sum_{i=1}^n X_i$ and $\|X\| = \sqrt{\sum_{i=1}^n \mathbb{E}[X_i^2]}$. Then we have exponential concentration for the upper tail
\[
\PP(X\geq \mathbb{E}[X]+\lambda) \leq \exp\left(-\frac{\lambda^2}{2(\|X\|^2 + M\lambda/3)}\right),~~~~\lambda>0.
\]
\end{theorem}

\begin{proof}
The first step is to use Markov inequality\footnote{
\begin{theorem}\label{Markov inequality}
For a random variable $X\geq 0$ and a positive number $a>0$, we have
\[
\PP(X\geq a) \leq \frac{1}{a}\mathbb{E}[X].
\]
\end{theorem}
\begin{proof}
Define random variable $Y$ to equal $a$ when $X\geq a$ and $0$ when $X<a$. Then, $0\leq Y\leq X$ and so $\mathbb{E}[X]\geq \mathbb{E}[Y] = a\PP(X\geq a)$.
\end{proof}
} for the exponentiated tail:
\begin{equation}\label{eq:exp_concentration_Markov_ineq}
	\PP(X\geq \mathbb{E}[X]+\lambda) = \PP\left(\exp(tX)\geq \exp(t\mathbb{E}[X]+t\lambda)\right) \leq \exp(-t\mathbb{E}[X]-t\lambda) \mathbb{E}[\exp(tX)].
\end{equation}
Now we treat $\mathbb{E}[\exp(tX)]$ carefully.
\begin{align*}
\mathbb{E}[\exp(tX)] =& \mathbb{E}\left[\exp(\sum_{i=1}^n tX_i)\right] = \mathbb{E} \left[\prod_{i=1}^n \exp(tX_i)\right] \\
=& \prod_{i=1}^n \mathbb{E}[\exp(tX_i)]~~~~\text{by independence of $X_i$}\\
=& \prod_{i=1}^n \mathbb{E}\left[\sum_{k=0}^{\infty} \frac{(tX_i)^k}{k!}\right] =\prod_{i=1}^n \mathbb{E}\left[1 + tX_i + \sum_{k=2}^{\infty} \frac{(tX_i)^k}{k!}\right]  \\
=& \prod_{i=1}^n \mathbb{E}\left[1 + tX_i + \frac{1}{2}t^2 X_i^2 \cdot g(tX_i)\right],
\end{align*}
where the function $g$ is defined as a power series
\[
g(y) = 2\sum_{k=2}^{\infty} \frac{y^{k-2}}{k!}.
\]
The purpose of this construction is to perform clever estimates on the quadratic terms involving second moments $\E[X^2]$. It is readily checked that $g$ has infinite radius of convergence, 
\[
g(y) = \frac{2(e^y-1-y)}{y^2},~~~~y\neq0.
\]
Moreover, $g(0)=1$, $g(y)\leq 1$ for $y<0$, $g(y)$ increases on $y\geq 0$, and if $|y|<3$, then
\[
g(y) = 2\sum_{k=2}^{\infty} \frac{y^{k-2}}{k!} \leq \sum_{k=2}^{\infty} \frac{y^{k-2}}{3^{k-2}} = \sum_{n=0}^{\infty} (y/3)^n = \frac{1}{1-y/3},
\]
where the first inequality follows from the fact that $k!\geq 2\cdot 3^{k-2}$.

We continue to estimate $\mathbb{E}[\exp(tX)]$
\begin{align*}
	\mathbb{E}[\exp(tX)] =& \prod_{i=1}^n \mathbb{E}\left[1+tX_i + \frac{1}{2}t^2X_i^2 \cdot g(tX_i)\right]= \prod_{i=1}^n \left[1+t\mathbb{E}[X_i] + \frac{1}{2}t^2\mathbb{E}[X_i^2  \cdot g(tX_i)]\right]\\
	\leq& \prod_{i=1}^n \left[1+t\mathbb{E}[X_i] + \frac{1}{2}t^2\mathbb{E}[X_i^2]  \cdot g(tM)\right],~~~~\text{by monotonicity of $g$}\\
	\leq& \prod_{i=1}^n \exp\left( t\mathbb{E}[X_i] + \frac{1}{2}t^2\mathbb{E}[X_i^2]  \cdot g(tM)\right),~~~~\text{by $1+v\leq e^v$}\\
	=& \exp\left(t\sum_{i=1}^n \mathbb{E}[X_i] + \frac{1}{2}t^2 \sum_{i=1}^n \mathbb{E}[X_i^2] g(tM) \right)\\
	=& \exp\left(t\mathbb{E}[X] + \frac{1}{2}t^2  \|X\|^2 g(tM)\right)\\
	\leq& \exp\left(t\mathbb{E}[X]+ \frac{1}{2}t^2 \|X\|^2 \frac{1}{1-tM/3}  \right),~~~~\text{choose $t$ with $tM<3$}.
\end{align*}
Plug this into inequality (\ref{eq:exp_concentration_Markov_ineq}) to obtain
\begin{align*}
	\PP(X\geq \mathbb{E}[X]+\lambda) \leq & \exp(-t\mathbb{E}[X]-t\lambda) \mathbb{E}[\exp(tX)]\\
	\leq& \exp\left(-t\lambda+ \frac{1}{2}t^2 \|X\|^2 \frac{1}{1-tM/3}  \right).
\end{align*}
To minimize $\left(-t\lambda+ \frac{1}{2}t^2 \|X\|^2 \frac{1}{1-tM/3}  \right)$, we put $t=\frac{\lambda}{\|X\|^2+M\lambda/3}$
 and note indeed $tM<3$. This gives
\begin{align*}
		\PP(X\geq \mathbb{E}[X]+\lambda) \leq& \exp\left(-\frac{\lambda^2}{2(\|X\|^2+M\lambda/3)}\right),
\end{align*}
as desired.
\end{proof}

By replacing $X$ with $-X$ in Theorem \ref{thm:exp_concentration+}, we obtain the following.

\begin{theorem}[Exponential concentration lower tail; \cite{CL06} Theorem 2.9]
\label{thm:exp_concentration-}
Suppose $X_i$, $i=1,\cdots,n$, are independent random variables with $X_i\geq- M$ for some $M>0$. Let $X:= \sum_{i=1}^n X_i$ and $\|X\| = \sqrt{\sum_{i=1}^n \mathbb{E}[X_i^2]}$. Then we have an exponential concentration for the lower tail
\[
\PP(X\leq \mathbb{E}[X]-\lambda) \leq \exp\left(-\frac{\lambda^2}{2(\|X\|^2 + M\lambda/3)}\right),~~~~\lambda>0.
\]
\end{theorem}

We obtain the classic Chernoff inequalities as consequences.

\begin{theorem}[Chernoff upper tail; \cite{CL06} Theorem 2.6] \label{thm:Chernoff+}
	\index{concentration inequality!Chernoff upper tail}
	Let $X_j$ be independent random variables satisfying $X_j\leq \mathbb{E}[X_j]+M$ for $j=1,\cdots,n$. Consider the sum $X:=\sum_{j=1}^n X_j$ with expectation $\mathbb{E}[X]=\sum_{j=1}^n \mathbb{E}[X_j]$ and variance $\mathrm{Var}[X]=\sum_{j=1}^n \mathrm{Var}[X_j]$. Then,
	$$\PP(X\geq \mathbb{E}[X]+\lambda) \leq \exp\left(-\frac{\lambda^2}{2(\mathrm{Var}[X]+ M \lambda/3)}\right).$$
\end{theorem}

\begin{proof}
	Define $X_i':= X_i - \mathbb{E}[X_i]$ and $X' := \sum_{i=1}^n X_i' = X - \mathbb{E}[X]$. Note $X_i'\leq M$ and $\|X'\|^2 = \sum_{i=1}^n \mathbb{E}[(X_i')^2] = \sum_{i=1}^n \mathbb{E}\left[(X_i - \mathbb{E}[X_i])^2\right]= \sum_{i=1}\mathrm{Var}(X_i)=\mathrm{Var}(X)$ by independence of $X_i$. Now apply Theorem \ref{thm:exp_concentration+} to obtain
	\[
	\PP(X\geq \mathbb{E}[X]+\lambda) = \PP(X'\geq \lambda) \leq \exp\left(-\frac{\lambda^2}{2(\|X'\|^2 + M\lambda/3)}\right) = \exp\left(-\frac{\lambda^2}{2(\mathrm{Var}(X) + M\lambda/3)}\right).
	\]
	This completes the proof.
\end{proof}

By choosing $M=0$ in Theorem \ref{thm:exp_concentration-}, we obtain the following.
\begin{theorem}[Chernoff lower tail; \cite{CL06} Theorem 2.7]\label{thm:Chernoff-}
	\index{concentration inequality!Chernoff lower tail}
	Let $X_j$ be nonnegative independent random variables for $j=1,\cdots,n$. Then the sum $X:=\sum_{j=1}^n X_j$ satisfies
	$$\PP(X\leq \mathbb{E}[X]-\lambda) \leq \exp\left(-\frac{\lambda^2}{2\sum_{j=1}^n \mathbb{E}[X_j^2]}\right).$$
\end{theorem}

\subsection{Concentration of degree in Chung-Lu model}
Applying these concentration inequalities with $\lambda=c\sqrt{w_i}$ for some carefully chosen $c>0$, Chung and Lu derived the concentration of the actual degree around its expected value.

\begin{lemma}[Joint concentration of actual degrees; \cite{CL06} Lemma 5.7] \label{thm:CLlemma5.7}
	\index{concentration inequality of actual degrees in Chung-Lu model}
	For a graph $G$ in $G(\brw)$, with probability $1-O(n^{-1/5})$ all vertices satisfy
	$$|k_{i}- w_i| \leq 2 (\sqrt{w_i \log n} + \log n).$$
\end{lemma}

\begin{proof}
	Recall $\mathbb{E}[k_i]=w_i$. Apply Chernoff lower tail Theorem \ref{thm:Chernoff-} to $k_i=\sum_{j=1}^n X_{ij}$ with $\lambda=c\sqrt{w_i}$, for some $c>0$ to be chosen later, to obtain
	\begin{align*}
		\PP(k_i\leq w_i - c\sqrt{w_i}) = \PP(k_i \leq \mathbb{E}[k_i] - \lambda) \leq \exp \left(-\frac{\lambda^2}{2\mathbb{E}[k_i]}\right)= e^{-c^2/2},
	\end{align*}
	where the inequality follows from the fact that the random variables $X_{ij}$ are Bernoulli; indeed,
	$$\sum_{j=1}^n \mathbb{E}[X_{ij}^2] = \sum_{j=1}^n \mathbb{E}[X_{ij}]= \mathbb{E}[k_i].$$
	
	Apply Chernoff upper tail Theorem \ref{thm:Chernoff+} to $k_i=\sum_{j=1}^n X_{ij}$ with $\lambda=c\sqrt{w_i}$, for some $c>0$ to be chosen later, to obtain
	\begin{align*}
		&\PP(k_i\geq w_i + c\sqrt{w_i}) = \PP(k_i \geq \mathbb{E}[k_i] + \lambda)\\
		\leq& \exp\left(-\frac{\lambda^2}{2(\mathbb{E}[k_i]+\lambda/3)}\right)= \exp\left(-\frac{c^2}{2(1+ c/(3\sqrt{w_i}))}\right).
	\end{align*}
	
	Now put
	$$c:=\begin{cases}
		2\sqrt{\log n},& \text{if~}w_i>\log n;\\
		\frac{2\log n}{\sqrt{w_i}},&\text{else}.
	\end{cases}$$
	For the lower tail bound, we have $c^2\leq 4\log n$ in both cases, and hence
	$$\exp(-c^2/2)\leq \exp(-2\log n)= n^{-2}.$$
	For the upper tail, when $w_i>\log n$, we have
	\begin{align*}
		\exp\left(-\frac{c^2}{2(1+ c/(3\sqrt{w_i}))}\right) =& \exp\left(-\frac{2\log n}{1+ 2\sqrt{\log n}/(3\sqrt{w_i})}\right)\\
		\leq & \exp\left(-\frac{2\log n}{1+ 2/3}\right)~~~~\text{using $w_i>\log n$}\\
		=& \exp(- (6/5) \log n) = n^{-6/5},
	\end{align*}
	and when $w_i\leq \log n$, we have
	\begin{align*}
		\exp\left(-\frac{c^2}{2(1+ c/(3\sqrt{w_i}))}\right) \leq& \exp\left(-\frac{c^2}{2(\log n/w_i+ c/(3\sqrt{w_i}))}\right) ~~~~\text{using $1\leq \log n /w_i$}\\
		&= \exp\left(- \frac{3}{5} c \sqrt{w_i}\right) = n^{-6/5}.
	\end{align*}
	Thus, with probability at least $1-n^{-2}-n^{-6/5}$, we have
	$$|k_i-w_i| \leq c\sqrt{w_i} \leq \max\{2\log n , 2\sqrt{w_i \log n}\}.$$
The proof is complete by removing the exceptional events of probability at most $n^{-2}-n^{-6/5}$ for each of the $n$ vertices.
\end{proof}

\section{Heterogeneous Mean field}

In a network, the degree \(k\) of a node represents the number of connections it has. Nodes are grouped into degree classes, and their infection dynamics are described collectively. The key idea 
is that infections are independent and follow a common field, thus, we can approximate
$$
\sum_{j}a_{ij} I_j \approx k \Theta
$$
where \(\Theta\) is the probability that a randomly chosen edge leads to an infected node. 
\(\Theta\) represents the "infectious pressure" in the network. It captures the probability that a randomly chosen edge leads to an infected node. This term connects the infection dynamics of individual nodes to the overall state of the network.
\(\Theta\) is expressed as:
\[
\Theta = \frac{\sum_k k P(k) I_k}{\langle k \rangle},
\]
where \(k P(k)\) is the contribution of degree \(k\) nodes to the network connectivity,
Notice that 
\(\Theta\) serves as a coupling term that reflects the degree-dependent infection dynamics at individual nodes and the aggregate effect of the network's structure and the infection states of all nodes.
It ensures that the infection dynamics at a node depend not only on its degree but also on the states of its neighbors and the network-wide distribution of infection.

High values of \(\Theta\) indicate a high proportion of infected nodes in the network, increasing the likelihood of transmission. On the other hand, low values of \(\Theta\) suggest fewer infected neighbors, reducing the likelihood of infection spread.

The dynamics for nodes of degree \(k\) then is
\[
\frac{dI_k}{dt} = \beta k (1 - I_k) \Theta - \mu I_k.
\]
This approximation is called the Heterogeneous mean field approximation.

{\it Equilibrium states}.
At steady state:
\[
I_k^* = \frac{\beta k \Theta}{\mu + \beta k \Theta}.
\]

Substitute \(I_k^*\) into the expression for \(\Theta\):
\[
\Theta = \frac{\sum_k k P(k) \frac{\beta k \Theta}{\mu + \beta k \Theta}}{\langle k \rangle},
\]
where \(P(k)\) is the degree distribution (probability that a randomly chosen node has degree \(k\)),
The threshold for epidemic persistence can be determined by analyzing the non-trivial solutions for \(\Theta\).

{\it Epidemic Threshold:} Lets consider the case where $\Theta\ll1$. Then linearizing 
around \(\Theta = 0\) we obtain
\[
\Theta = \frac{\sum_k k P(k)  r_0 k \Theta \left(1 - r_0 k \Theta + O(\Theta^2)\right)}{\langle k \rangle},
\]
since $\Theta $ is nonzero we obtain the leading approximation
\[
r_0 = \frac{\langle k \rangle}{\langle k^2 \rangle} 
\]
for the epidemic threshold. Let's see some applications

\bibliographystyle{alpha}
\bibliography{references}.

\end{document}

%% file: first.tex
\begin{titlepage}

\vspace{2cm}
\begin{center}
{\Huge {Epidemics models in Networks}}\\
\vspace{8.1 cm}
{\large \bf  Tiago Pereira} \\
Instituto de Ciências Matemáticas e Computação   \\
Universidade de São Paulo \\ 
tiago@icmc.usp.br
\vspace{1.3 cm}


\end{center}
\end{titlepage}
\thispagestyle{empty} 





%% file: preface.tex
%
%

\preface


~ ~ These lectures are based on material which was presented in the 2025 Summer 
school at Fundação Getulio Vargas. The aim of this series is to introduce graduate students with a little background in the field of dynamical systems and network theory to epidemic models. 

Our goal is to give a succinct and self-contained description of the models. We assume that the reader has basic knowledge of linear algebra and the theory of differential equations (we include the main results necessary to follow the lectures in the appendix, though). 

During the course, I was able to cover more material by blending expository lectures with live simulations using MATLAB and conducting on-the-fly calculations and discussions with ChatGPT. However, when I put those in these notes, I felt the material was hard to follow and was becoming excessively long. Thus, I have kept only the necessary parts.  Also, some of these lectures were given by Eddie Nijholt (USP), Zheng Bian (UNSW), and Matthias Wolfrum (Weierstrass). Needless to say, I am in debt with them. These lectures also benefited from discussions with Claudio Struchiner, Guilherme Goedert, Luiz Max, and Cristiana Couto. This work was partially supported by CNPq, FAPESP, and Instituto Serrapilheira. We also acknowledge the support of the Humboldt Foundation via the Bessel Fellowship at the Weierstrass Institute for Analysis in Berlin.

\vspace{\baselineskip}
\begin{flushright}\noindent
S\~ao Carlos,\hfill {\it Tiago Pereira}\\
December  2025\hfill {\it }\\
\end{flushright}